\documentclass[notitlepage,11pt,reqno]{amsart}
\usepackage{amssymb,amsmath,amsthm,mathrsfs,xspace,multicol,stmaryrd}
\usepackage[mathscr]{eucal}
\usepackage{amscd}
\usepackage[all]{xy}
\usepackage[breaklinks=true]{hyperref}
\usepackage{url}

\newcommand{\blankbrac}{[ \cdot,\cdot ]}
\newcommand{\brac}[2]{\left \{ #1,#2 \right\}}

\newcommand{\sbrac}[2]{\left [ #1,#2 \right ]}
\newcommand{\lbrac}[2]{\left [ #1,#2 \right ]_{L}}
\newcommand{\lpbrac}[2]{\left [ #1,#2 \right ]_{L'}}

\newcommand{\cbrac}[2]{\left [ #1,#2 \right ]_{C}}
\newcommand{\tcbrac}[2]{\left [ #1,#2 \right ]_{C}}
\newcommand{\abrac}[2]{\left [#1,#2 \right]_{A}}
\newcommand{\aibrac}[2]{\left [#1,#2 \right]_{i}}
\newcommand{\tabrac}[2]{\left [#1,#2 \right]_{A}}
\newcommand{\dbrac}[2]{\left \llbracket #1,#2 \right \rrbracket_{C}}
\newcommand{\tdbrac}[2]{\left \llbracket #1,#2 \right  \rrbracket_{C}}

\newcommand{\innerprodp}[2]{\bigl \langle #1  ,  #2 \bigr \rangle^{+}}
\newcommand{\innerprodm}[2]{\bigl \langle #1 ,  #2 \bigr \rangle^{-} }

\newcommand{\ham}{\Omega^{1}_{\mathrm{Ham}}(M)}
\newcommand{\X}{\mathfrak{X}}
\newcommand{\Xham}{\mathfrak{X}_{\mathrm{Ham}}(M)}

\newcommand{\cinf}{C^{\infty}}

\newcommand{\lie}[2]{\mathcal{L}\,_{v_{#1}}\, #2}
\renewcommand{\L}{\mathcal{L}}
\newcommand{\U}{\mathcal{U}}
\newcommand{\ip}[1]{\iota_{v_{#1}}}

\newcommand{\sh}[1]{\underline{#1}}
\newcommand{\st}[1]{\mathcal{#1}}

\newcommand{\cOmega}{\Omega_{\mathrm{cl}}}

\newcommand{\R}{\mathbb{R}}
\newcommand{\C}{\mathbb{C}}
\newcommand{\Z}{\mathbb{Z}}

\newcommand{\innerprod}[2]{\langle #1,#2 \rangle}
\newcommand{\g}{\mathfrak{g}}

\renewcommand{\u}{\mathfrak{u}}

\newcommand{\tensor}{\otimes}

\newcommand{\maps}{\colon}
\renewcommand{\i}{\sqrt{-1}}

\newcommand{\half}{\frac{1}{2}}

\newcommand{\epi}{\twoheadrightarrow}

\DeclareMathOperator{\Aut}{\mathrm{Aut}}
\DeclareMathOperator{\Hom}{\mathrm{Hom}}

\DeclareMathOperator{\id}{\mathrm{id}}
\DeclareMathOperator{\im}{\mathrm{im}}
\DeclareMathOperator{\cp}{\mathrm{c.\!p}}

\theoremstyle{plain}
\newtheorem{theorem}{Theorem}[section]
\newtheorem{prop}[theorem]{Proposition}
\newtheorem{lemma}[theorem]{Lemma}
\newtheorem{corollary}[theorem]{Corollary}
\newtheorem{definition}[theorem]{Definition}

\theoremstyle{remark}

\newtheorem{example}{Example}

\begin{document}

\title[2-plectic geometry, Courant algebroids] {2-plectic geometry,
  Courant algebroids, and categorified prequantization}
\author{Christopher L.\ Rogers} \email{\texttt{chris@math.ucr.edu}}
\address{Department of Mathematics, University of California,
  Riverside, California 92521, USA} \date{\today}
\curraddr{Mathematisches Institut, Georg-August-Universiti\"{a}t
  G\"{o}ttingen, Bunsenstr. 3-5, D-37073, G\"{o}ttingen, DE}
\thanks{This work was supported by a Junior Research Fellowship from
  the Erwin Schr\"{o}dinger International Institute for Mathematical
  Physics and FQXi grant RFP2-08-04.}


\begin{abstract}
  A 2-plectic manifold is a manifold equipped with a closed
  nondegenerate 3-form, just as a symplectic manifold is equipped with
  a closed nondegenerate 2-form.  In 2-plectic geometry one finds the higher
  analogues of many structures familiar from symplectic geometry.  For
  example, any 2-plectic manifold has a Lie 2-algebra consisting of
  smooth functions and Hamiltonian 1-forms.  This is equipped with a
  Poisson-like bracket which only satisfies the Jacobi identity up
  to `coherent chain homotopy'. Over any 2-plectic manifold is a vector
  bundle equipped with extra structure called an exact Courant
  algebroid.  This Courant algebroid is the 2-plectic analogue of a
  transitive Lie algebroid over a symplectic manifold. Its space of
  global sections also forms a Lie 2-algebra. We show that this Lie
  2-algebra contains an important sub-Lie 2-algebra which is
  isomorphic to the Lie 2-algebra of Hamiltonian 1-forms.
  Furthermore, we prove that it is quasi-isomorphic to a central
  extension of the (trivial) Lie 2-algebra of Hamiltonian vector
  fields, and therefore is the higher analogue of the well-known
  Kostant-Souriau central extension in symplectic geometry. We
  interpret all of these results within the context of a categorified
  prequantization procedure for 2-plectic manifolds. In doing so, we
  describe how $U(1)$-gerbes, equipped with a connection and curving,
  and Courant algebroids are the 2-plectic analogues of principal
  $U(1)$ bundles equipped with a connection and their associated
  Atiyah Lie algebroids.
  \end{abstract}

\maketitle

\section{Introduction}
\label{introduction}
A multisymplectic manifold is a smooth manifold equipped with a
closed, nondegenerate form of degree $\geq 2$.
In this paper, we call a manifold `$n$-plectic' if the form has degree $(n+1)$. These
manifolds naturally arise in certain covariant Hamiltonian formalisms
for classical field theory \cite{Helein,Kijowski,RomanRoy:2005en}.
In these formalisms, one describes a $(n+1)$-dimensional field theory
by using a finite-dimensional $n$-plectic manifold as
a `multi-phase space' instead of an infinite-dimensional
phase space. The $n$-plectic form can be used to define a
system of partial differential equations which are the analogue of
Hamilton's equations in classical mechanics. The solutions to these
equations correspond to particular submanifolds of the multi-phase space
that encode the value of the field at each point in space-time as
well as the values of its time and spatial derivatives.

Other formalisms, such as higher gauge theory
\cite{BaezSchreiber:2005,Bartels:2004,Schreiber:2005}, suggest that
structures found in classical mechanics can be generalized by using
higher category and homotopy theory and then applied to the study of
field, string, and brane theories.  Motivated by these ideas, we
hypothesized in our previous work with Baez and Hoffnung
\cite{Baez:2008bu} that the higher analogues of well-known algebraic
and geometric structures on symplectic manifolds should naturally
arise on $n$-plectic manifolds. Algebraically, this is indeed true.
Just as a symplectic structure makes the ring of smooth functions a
Poisson algebra, an $n$-plectic structure gives a Lie $n$-algebra on a
$n$-term chain complex consisting of differential $p$-forms for $0
\leq p \leq n-2$ and certain $(n-1)$-forms which we call Hamiltonian
\cite{Rogers:2010nw}.  A Lie $n$-algebra (or $n$-term
$L_{\infty}$-algebra) is a higher analogue of a differential graded
Lie algebra. It consists of a graded vector space concentrated in
degrees $0,\ldots,n-1$ equipped with a collection of
skew-symmetric $k$-ary brackets, for $1 \leq k \leq n+1$, that satisfy
a generalized Jacobi identity \cite{Lada-Markl,LS}. In particular, the
$k=2$ bilinear bracket behaves like a Lie bracket that only satisfies
the ordinary Jacobi identity up to higher coherent chain homotopy.
When $n=1$, the relevant Lie $1$-algebra is just the underlying Lie
algebra of the usual Poisson algebra. When $n=2$, we obtain a Lie
2-algebra whose underlying 2-term chain complex consists of smooth
functions and Hamiltonian 1-forms.

Now let us consider the geometric picture.  Interesting geometric
structures appear, in particular, on prequantizable symplectic
manifolds i.e.\ those manifolds $(M,\omega)$ with the property that
the integral of the symplectic form $\omega$ over any closed oriented
2-surface is an integer multiple of $2 \pi \i$. In this case, there
exists a principal $U(1)$-bundle $P \stackrel{\pi}{\to} M$ over the
manifold equipped with a connection whose curvature is $\pi$-related
to the symplectic form. Equivalently, in terms of cohomology, the
symplectic structure gives a representative of a degree 2 class in
integer-valued cohomology, while the data encoding the principal
bundle with connection give a representative of a degree 1 class in
Deligne cohomology. Deligne cohomology can be interpreted as a
refinement of the more familiar $U(1)$-valued \v{C}ech cohomology.  In
degree 1, it classifies not just principal $U(1)$-bundles, but
principal $U(1)$-bundles equipped with connection.

Another geometric structure, called the Atiyah algebroid, is also
present on a prequantized symplectic manifold. The Atiyah algebroid is an example
of a Lie algebroid: roughly, a vector bundle $A \to M$ equipped with a
bundle map to the tangent bundle of $M$, and a Lie algebra structure
on its space of global sections.  The total space of the Atiyah
algebroid is the quotient $A=TP/U(1)$, where $P \to M$ is the
aforementioned principal $U(1)$-bundle. Sections of $A$ are
$U(1)$-invariant vector fields on $P$. A connection on $P$ is
equivalent to a splitting of the short exact sequence
\[
0 \to \R \times M \to A\stackrel{\pi_{\ast}}{\to} TM \to 0
\]
where the map $\R \times M \to A$ corresponds to identifying the
vertical subspace of $T_pP$ with the Lie algebra $\u(1) \cong
\R$. Those sections of $A$ which preserve the connection (or
splitting) form a Lie subalgebra that is isomorphic to the Poisson
algebra.  This implies that there is a well-defined action of the
Poisson algebra on the $\C$-valued functions on $P$. Compactly
supported global sections of the line bundle associated to $P$ form a
pre-Hilbert space and can be identified with $U(1)$-homogeneous
$\C$-valued functions on $P$ of degree $-1$.  In this way one obtains
a faithful representation, or a quantization, of the Poisson algebra
by linear operators on a Hilbert space. Moreover, if the symplectic
manifold is connected, then the Poisson algebra gives what is known as
the Kostant-Souriau central extension of the Lie algebra of
Hamiltonian vector fields \cite{Kostant:1970}. The symplectic form,
evaluated at a point, gives a representative of the degree 2 class in
the Lie algebra cohomology of the Hamiltonian vector fields (with
values in the trivial representation) corresponding to this
extension. The fact that this central extension is quantized, rather
than the Hamiltonian vector fields themselves, is the reason why the
concept of `phase' is introduced in quantum mechanics.

The process described above is known as prequantization
\cite{Kostant:1970}. It is the first step towards geometrically
quantizing a symplectic manifold \cite{Kostant:1970, Souriau:1967}. We
are interested in the higher analogues of the geometric structures
described above. Indeed, the geometric quantization of what we call an
$n$-plectic manifold remains a long-standing open problem. In this
paper, we focus particularly on the prequantization of 2-plectic
manifolds, since this is the first really new case of $n$-plectic
geometry. Hence, we study prequantized 2-plectic manifolds and
the 2-plectic analogues of principal $U(1)$-bundles, Atiyah
algebroids, and the Kostant-Souriau central extension.

In analogy with the symplectic case, a 2-plectic manifold $(M,\omega)$
is prequantizable if the integral of the 2-plectic form $\omega$ over
any closed oriented 3-surface is an integer multiple of $2 \pi
\i$. Hence, the 2-plectic structure gives a
representative of a degree 3 class in integer-valued cohomology. This
degree 3 class corresponds to a (not necessarily unique) degree 2
class in Deligne cohomology. It is well-known that a geometric object
which realizes this degree 2 class is a $U(1)$-gerbe over $M$ equipped
with a connection and curving whose 3-curvature is $\omega$
\cite{Brylinski:1993}. Roughly, a $U(1)$-gerbe is a stack (or sheaf of
groupoids) over $M$ that is locally isomorphic to the stack of
$U(1)$-bundles over $M$.  Just as a connection on a principal
$U(1)$-bundle is equivalent to specifying local 1-forms on $M$
satisfying a cocycle condition, the connection and curving on a
$U(1)$-gerbe correspond to specifying local 1-forms and 2-forms on $M$
satisfying a pair of cocycle conditions.  So, by going from symplectic
geometry to 2-plectic geometry, we are replacing sets of local
sections of a principal bundle (i.e.\ sheaves) by categories of
principal bundles defined over open sets (i.e.\ stacks). Therefore the
prequantization of a 2-plectic manifold is in some sense `categorified
prequantization'.

What is the 2-plectic analogue of the Atiyah algebroid? We answer this
question by considering a more general problem: understanding the
relationship between 2-plectic geometry and the theory of Courant
algebroids. Roughly, a Courant algebroid is a vector bundle that
generalizes the structure of a Lie algebroid equipped with a symmetric
nondegenerate bilinear form on the fibers. They were first used by
Courant \cite{Courant} to study generalizations of pre-symplectic and
Poisson structures in the theory of constrained mechanical
systems. Curiously, many of the ingredients found in 2-plectic
geometry are also found in the theory of `exact' Courant
algebroids. An exact Courant algebroid is a Courant algebroid whose
underlying vector bundle $C \to M$ is an extension of the tangent
bundle by the cotangent bundle:
\[
0 \to T^{\ast}M \to C \to TM \to 0.
\] 
In a letter to Weinstein, \v{S}evera \cite{Severa1} described how
exact Courant algebroids arise in 2 dimensional variational problems
(e.g.\ bosonic string theory) and showed that they are classified up
to isomorphism by the degree 3 de Rham cohomology of $M$. From any
closed 3-form on $M$, one can explicitly construct an exact Courant
algebroid equipped with an `isotropic' splitting of the above short
exact sequence using local 1-forms and 2-forms that satisfy cocycle
conditions \cite{Bressler-Chervov,Hitchin:2004ut,Gualtieri:2007}.

Obviously, \v{S}evera's classification implies that every 2-plectic
manifold $(M,\omega)$ gives a unique exact Courant algebroid (up to
isomorphism) $C$ whose class is represented by the 2-plectic
structure. However, there are more interesting similarities between
2-plectic structures and exact Courant algebroids.  Roytenberg and
Weinstein \cite{Roytenberg-Weinstein} showed that the bracket on the
space of global sections of a Courant algebroid induces an
$L_{\infty}$ structure. If we are considering an exact Courant
algebroid, then the global sections can be identified with vector
fields and 1-forms on the base space. Roytenberg and Weinstein's
results imply that these sections, when combined with the smooth
functions on the base space, form a Lie 2-algebra
\cite{Roytenberg_L2A}. Moreover, the ``higher brackets'' of the Lie 2-algebra
encode a closed 3-form representing the \v{S}evera class
\cite{Severa-Weinstein}.

The first result we present in this paper is that there exists a Lie
2-algebra morphism which embeds the Lie 2-algebra of Hamiltonian
1-forms on a 2-plectic manifold $(M,\omega)$ into the Lie 2-algebra of
global sections of the corresponding exact Courant algebroid
$C$ equipped with an isotropic splitting.  Moreover, this
morphism gives an isomorphism between the Lie 2-algebra of Hamiltonian
1-forms and the sub Lie 2-algebra consisting of sections of
$C$ which preserve the splitting via a particular kind of
adjoint action. This result holds without any integrality condition on
the 2-plectic structure. However, its meaning becomes clear in the
context of prequantization: it is the higher analogue of the isomorphism
between the underlying Lie algebra of the Poisson algebra on a
prequantized symplectic manifold and the Lie sub-algebra of sections
of the Atiyah algebroid that preserve the connection on the associated
principal bundle. Hence, we see that the 2-plectic analogue of the
Atiyah algebroid associated to a principal $U(1)$-bundle is an exact
Courant algebroid associated to a $U(1)$-gerbe. This idea that exact
Courant algebroids are `higher Atiyah algebroids' has been discussed
previously in the literature
\cite{Bressler-Chervov,Gualtieri:2007}. However, this is the first
time the analogy has been understood using Lie $n$-algebras within the
context of multisymplectic geometry.

The second result presented here involves identifying the 2-plectic
analogue of the Kostant-Souriau central extension and therefore the
source of `phase' in categorified prequantization. On a 2-plectic
manifold, associated to every Hamiltonian 1-form is a Hamiltonian
vector field. These vector fields form a Lie algebra which we can view
as a trivial Lie 2-algebra whose underlying chain complex is
concentrated in degree 0, and whose bracket satisfies the Jacobi
identity on the nose. For any 1-connected (i.e.\ connected and simply
connected) 2-plectic manifold, we show that the Lie 2-algebra of
Hamiltonian 1-forms is quasi-isomorphic to a `strict central
extension' of the trivial Lie 2-algebra of Hamiltonian vector fields
by the abelian Lie 2-algebra $\R \to 0$. This abelian Lie 2-algebra is
known as $b\u(1)$. Furthermore, we show this extension corresponds to
a degree 3 class in the Lie algebra cohomology of the Hamiltonian
vector fields with values in the trivial representation. In analogy
with the symplectic case, a 3-cocycle representing this class can be
constructed by using the 2-plectic form. It follows from the
aforementioned results relating a 2-plectic manifold $(M,\omega)$ to
the Courant algebroid $C$ that the sub Lie 2-algebra of sections of
$C$ that preserve the splitting is also quasi-isomorphic to this
central extension, and can be interpreted as the quantization of the
Lie 2-algebra of Hamiltonian 1-forms. Phases originate from the
presence of $b\u(1)$, which integrates to an important Lie 2-group
called $BU(1)$.

In the next section, we briefly review the construction of transitive
Lie algebroids on symplectic manifolds and describe an embedding of
the Poisson algebra into the Lie algebra of sections of the
algebroid. We recall some basic facts concerning Deligne cohomology
and then consider prequantized symplectic manifolds. We emphasize the
role played by the Atiyah algebroid in prequantization and the
construction of the Kostant-Souriau central extension.  The remainder
of the paper is devoted to the 2-plectic analogue.  In Sections
\ref{2plectic_sec} and \ref{courant_sec} we introduce 2-plectic
manifolds and Courant algebroids as well as review \v{S}evera's
classification theorem for exact Courant algebroids. Section
\ref{geometric} contains a description of the geometric relationship
between 2-plectic manifolds and exact Courant algebroids. After
reviewing Lie 2-algebras in Section \ref{L2A_sec}, we present the
algebraic relationship between 2-plectic and Courant in Section
\ref{algebraic}. In Section \ref{2plectic_quant}, we introduce
prequantized 2-plectic manifolds and describe how the exact Courant
algebroid plays the role of a higher Atiyah algebroid. We then present
in Section \ref{2plectic_extend} the 2-plectic analogue of the
Kostant-Souriau central extension. We assume the reader is comfortable
with basic results in symplectic geometry and geometric quantization,
but not necessarily familiar with Deligne cohomology, gerbes, Courant
algebroids, or Lie 2-algebras. Therefore, our presentation of these
topics is mostly self-contained.

\section{Lie algebroids, symplectic manifolds, and
  prequantization} \label{symplectic} 
\subsection{Lie algebroids from closed 2-forms} \label{symp_algebroid}
We begin by reviewing the construction of a Lie algebroid which
ultimately will describe how phases arise in the prequantization of
symplectic manifolds. A section of this Lie algebroid is a vector
field on the base manifold together with a `phase', or more precisely,
a real-valued function.

Recall that a \textbf{Lie algebroid} \cite{Mackenzie:1987} over a manifold $M$ is a real
vector bundle $A \to M$ equipped with a bundle map (called the anchor) $\rho \maps A \to
TM$, and a Lie algebra bracket $\blankbrac_{A} \maps \Gamma(A) \tensor
\Gamma(A) \to \Gamma(A)$ such that the induced map
\[
\Gamma(\rho) \maps \Gamma(A) \to \X(M)
\]
is a morphism of Lie algebras, and for all $f \in \cinf(M)$ and
$e_{1},e_{2} \in \Gamma(A)$ we have
the Leibniz rule
\[
\abrac{e_{1}}{f e_{2}} = f \abrac{e_{1}}{e_{2}} + \rho(e_{1})(f)e_{2}. 
\]
A Lie algebroid with surjective anchor map is called a
\textbf{transitive Lie algebroid}. 

The main ideas of the following construction are presented in Sec.\ 17
of Cannas da Silva and Weinstein \cite{CdS-Weinstein:1999}. We provide
the details here in order to compare to the 2-plectic case in Sec.\
\ref{geometric}. Let $(M,\omega)$ be a manifold equipped with a closed
2-form, e.g.\ a pre-symplectic manifold. By a \textbf{trivialization}
of $\omega$, we mean a cover $\{U_{i}\}$ of $M$, equipped with 1-forms
$\theta_{i} \in \Omega^{1}(U_{i})$, and smooth functions $g_{ij} \in
\cinf(U_{i} \cap U_{j})$ such that
\begin{align} 
 \omega \vert_{U_{i}} = d\theta_{i} \\
(\theta_{j}-\theta_{i}) \vert_{U_{ij}} = dg_{ij}, \label{cocycle1}
\end{align}
where $U_{ij}=U_{i} \cap U_{j}$.
Every manifold admits a good cover (i.e.\ a cover where all non-empty finite
intersections $U_{i_{i}\ldots i_{k}}=U_{i_{1}} \cap \ldots \cap
U_{i_{k}}$ are contractible), hence every closed 2-form admits a trivialization.   
Given such a trivialization of $\omega$, we can construct a transitive Lie
algebroid over $M$. Over each $U_{i}$ we consider the Lie algebroid
\[
A_{i}=TU_{i} \oplus \R \to U_{i},
\]
with bracket 
\[
\aibrac{v_{1} + f_{1}}{v_{2} + f_{2}} = [v_{1},v_{2}] + v_{1}(f_{2}) -
v_{2}(f_{1}) 
\]
for all  $v_{i} + f_{i} \in \X(U_{i}) \oplus \cinf(U_{i})$, and anchor $\rho$ given
by the projection onto $TU_{i}$.
From the 1-forms $dg_{ij} \in \Omega^{1}(U_{ij})$, we can construct
transition functions
\begin{align*}
G_{ij} \maps U_{ij} \to GL(n+1),\\
G_{ij}(x)= 
\begin{pmatrix}
1 & 0\\
dg_{ij} \vert_{x} & 1
\end{pmatrix},
\end{align*}
which act on a point $v_{x} + r \in A_{i} \vert_{U_{ij}}$ by 
\[
G_{ij}(x)(v_{x} + r) = v_{x} + r + dg_{ij}(v_{x}).
\]
Clearly, each $G_{ij}$ satisfies the cocycle conditions on $U_{ijk}$ by virtue of Eq.\ \ref{cocycle1}.
Therefore, we have over $M$ the vector bundle
\[
A= \coprod_{x \in M} T_{x}U_{i} \oplus \R / \sim,
\]
where the equivalence is defined via the functions $G_{ij}$ in
the usual way. For any sections $v_{i} + f_{i}$ of $A_{i}
\vert_{U_{ij}}$, a direct calculation shows that
\[
\aibrac{G_{ij}(v_{1} + f_{1})}{G_{ij}(v_{2} + f_{2})} = G_{ij}([v_{1},v_{2}] + v_{1}(f_{2}) -
v_{2}(f_{1})).
\]
Hence the local bracket descends to a well-defined bracket
$\tabrac{\cdot}{\cdot}$ on the quotient. Henceforth, 
$(A,\tabrac{\cdot}{\cdot},\rho)$ will denote this transitive Lie algebroid
associated to the closed 2-form $\omega$.

It's easy to see that the above Lie algebroid is an extension of the
tangent bundle
\[
0 \to M \times \R \to A \stackrel{\rho}{\to} TM \to 0.
\]
Moreover, the 1-forms $\theta_{i} \in \Omega^{1}(U_{i})$ induce a splitting
\[
s \maps TM \to A
\]
of the above sequence defined as
\begin{equation} \label{A_splitting}
s(v_{x})= v_{x} - \theta_{i}(v_{x}), \quad \forall ~ v_{x} \in TU_{i}.
\end{equation}
By a slight abuse of notation, we denote the horizontal lift
$\Gamma(s) \maps \X(M) \to \Gamma(A)$ also by $s$. Hence
every section $e \in \Gamma(A)$ is of the form $e=s(v) +f$,
for some $v\in \X(M)$ and $f \in \cinf(M)$. Using the local definition
of the splitting and the fact that $\omega \vert_{U_{i}}=d\theta_{i}$,
a direct calculation shows that
\begin{equation} \label{tabrac_def}
\tabrac{s(v_{1}) + f_{1}}{s(v_{2}) +
  f_{2}} = s \bigl ([v_{1},v_{2}] \bigr) +
v_{1}(f_{2})- v_{2}(f_{1}) 
-\ip{2}\ip{1}\omega,
\end{equation}
for all sections $s(v_{i}) + f_{i}$. The failure of
the splitting $s \maps TM \to A$ to preserve the Lie bracket on
sections is measured by the 2-form $\omega$:
\[
[s(v_{1}),s(v_{2})]_{A} = s([v_{1},v_{2}]) - \omega(v_{1},v_{2}),
\quad \forall v_{1},v_{2} \in \X(M).
\]

It is a simple exercise to show that a different choice of
trivialization gives a Lie algebroid equipped with a splitting that
is isomorphic to $A$ equipped with the splitting given in Eq.\ \ref{A_splitting}.

\subsection{The Poisson algebra}\label{symp_preserve} 
Let $(M,\omega)$ be a symplectic
manifold. Here $\brac{f}{g}=\omega(v_{f},v_{g})$ denotes the Poisson
bracket on smooth functions. The vector field $v_{f}$, satisfying 
the equality $df=-\ip{f}\omega$, is the unique Hamiltonian vector
field corresponding to the function $f$.  We denote the Lie algebra of
Hamiltonian vector fields by $\Xham$.  Let
$(A,\tabrac{\cdot}{\cdot},\rho)$ be the Lie algebroid
associated to $\omega$ and $s \maps TM \to A$ be the
splitting defined in Eq.\ \ref{A_splitting}. We are interested in a particular Lie
sub-algebra of $\Gamma(A)$  acting on the subspace
$s(\X(M)) \subseteq \Gamma(A) $ via the adjoint action.
\begin{definition} \label{preserve_split_def}
A section $a=s(v) + f \in \Gamma(A)$ {\bf preserves the splitting} $s \maps TM
\to A$ iff $\forall v' \in \X(M)$ 
\[
\bigl [a,s(v') \bigr]_{A} = s([v,v']).
\]
The subspace of sections that preserve the splitting is 
denoted as $\Gamma(A)^{s}$.
\end{definition}

\begin{prop} \label{preserve_lie_alg}
$\Gamma(A)^{s}$ is a Lie subalgebra of $\Gamma(A)$.
\end{prop}
\begin{proof}
Follows directly from the fact that the bracket on $\Gamma(A)$ and the
bracket on $\X(M)$ both satisfy the Jacobi identity.
\end{proof}
It is easy to show that a section $s(v)+f$ preserves the splitting if
and only if $v=v_{f}$. In fact: 
\begin{prop} \label{lie_alg_iso}
The underlying Lie algebra of the Poisson algebra
$\bigl(\cinf(M),\brac{\cdot}{\cdot} \bigr)$ is isomorphic to the Lie algebra
$\bigl(\Gamma(A)^{s},\tabrac{\cdot}{\cdot} \bigr)$.
\end{prop}
\begin{proof}
For any vector field $v' \in \X(M)$, it follows from Eq.\
\ref{tabrac_def} that we have $\tabrac{s(v)+f}{s(v')}=s([v,v'])$ if
and only if
\[
v'(f)+\omega(v,v')=0,
\]
and hence $df=-\iota_{v}\omega$.
Therefore the injective map
\[
\phi \maps \cinf(M) \to \Gamma(A)^{s}, \quad \phi(f)=s(v_{f})
+ f
\]
is also surjective. If  $v_{f}$ and $v_{g}$ are Hamiltonian vector fields corresponding to the functions
$f$ and $g$, respectively then
\begin{align*}
\left [ \phi(f), \phi(g) \right ]_{A} &=\left [s(v_{f}) + f, s(v_{g}) + g
\right ]_{A} \\
&=  s([v_{f},v_{g}]) + \bigl (v_{f}(g) - v_{g}(f) \bigr) 
- \ip{g}\ip{f}\omega \\
&=s([v_{f},v_{g}]) + \omega(v_{f},v_{g})\\
&= \phi(\brac{f}{g}).
\end{align*}
\end{proof}

\subsection{Deligne cohomology} \label{Deligne_sec}
We now briefly review some basic facts concerning smooth Deligne
cohomology. We will mainly use this as a convenient language for dealing
with geometric objects, such as principal $U(1)$-bundles or
$U(1)$-gerbes, equipped with extra structure. Our presentation follows
Sec.\ 3 of Carey, Johnson, and Murray \cite{Carey:2004}. For more details, we refer
the reader to the book by Brylinski \cite{Brylinski:1993}.

Let $\sh{U(1)}$ and $\Omega^{k}$ denote the sheaves of smooth $U(1)$-valued functions
and differential $k$-forms,
respectively, on a manifold $M$. Consider the exact sequence of sheaves $D_{p}^{\bullet}$:
\[
\sh{U(1)} \stackrel{d \log }{\to} \Omega^{1} \stackrel{d}{\to}
\cdots \stackrel{d}{\to} \Omega^{p}, \quad p \geq 1.
\]
Define the Deligne cohomology $H^{\bullet}(M,D_{p}^{\bullet})$ to be the \v{C}ech
hyper-cohomology of $D_{p}^{\bullet}$. This is the total cohomology
of the double complex:
\[
\xymatrix{
\vdots & \vdots& \vdots && \vdots\\
\sh{U(1)}(\U^{[2]}) \ar[u]^{\delta} \ar[r]^{d\log}& \Omega^{1}(\U^{[2]})
\ar[u]^{\delta} \ar[r]^{d} &
\Omega^{2}(\U^{[2]}) \ar[u]^{\delta} \ar[r]^{d}&\cdots
\ar[r]^{d} & \ar[u]^{\delta} \Omega^{p}(\U^{[2]}) \\
\sh{U(1)}(\U^{[1]})\ar[u]^{\delta} \ar[r]^{d\log}& \Omega^{1}(\U^{[1]})
\ar[u]^{\delta} \ar[r]^{d} &
\Omega^{2}(\U^{[1]}) \ar[u]^{\delta} \ar[r]^{d}&\cdots
\ar[r]^{d} & \ar[u]^{\delta} \Omega^{p}(\U^{[1]}) \\
\sh{U(1)}(\U^{[0]}) \ar[u]^{\delta} \ar[r]^{d \log}& \Omega^{1}(\U^{[0]})
\ar[u]^{\delta} \ar[r]^{d} &
\Omega^{2}(\U^{[0]}) \ar[u]^{\delta} \ar[r]^{d}&\cdots
\ar[r]^{d} & \ar[u]^{\delta}\Omega^{p}(\U^{[0]}) \\
}
\]
where $\U=\{U_{i}\}$ is a good cover of $M$, $\delta$ is the usual
\v{C}ech co-boundary operator, and $\sh{U(1)}(\U^{[n]})$ and 
$\Omega^{k} (\U^{[n]})$ denote the abelian groups
\[
\sh{U(1)}(\U^{[n]})=\prod_{i_{0} \neq i_{1} \neq \cdots \neq i_{n}}
\sh{U(1)} \left (U_{i_{0}} \cap U_{i_{2}} \cdots \cap U_{i_{n}} \right),
\]
\[
\Omega^{k} (\U^{[n]})=\prod_{i_{0} \neq i_{1} \neq \cdots \neq i_{n}}
\Omega^{k} \left (U_{i_{0}} \cap U_{i_{2}} \cdots \cap U_{i_{n}} \right).
\]
We will focus on the groups $H^{p}(M,D_{p}^{\bullet})$.
They can be thought of as a refinement of the usual
\v{C}ech cohomology groups $H^{\bullet}(M,\sh{U(1)})$. In particular, there
is a surjection
\[
H^{p}(M,D_{p}^{\bullet}) \epi  H^{p}(M,\sh{U(1)}). 
\]
Hence, via the usual isomorphism $H^{p}(M,\sh{U(1)}) \cong
H^{p+1}(M,\Z)$, we have a surjection
\begin{equation} \label{chern}
c \maps H^{p}(M,D_{p}^{\bullet})\epi  H^{p+1}(M,\Z). 
\end{equation}
If $[\xi] \in H^{p}(M,D_{p}^{\bullet})$, then $c([\xi])$ is called
the \textbf{Chern class} of $[\xi]$. 

There is also a map of complexes
\[
\xymatrix{
\sh{U(1)} \ar[d] \ar[r]^{d \log}& \Omega^{1}
\ar[d] \ar[r]^{d} &
\Omega^{2} \ar[d] \ar[r]^{d}&\cdots
\ar[r]^{d} & \ar[d]^{d}\Omega^{p} \\
0 \ar[r] & 0 \ar[r]&  0 \ar[r] & \cdots \ar[r] & \Omega^{p+1}
}
\]
where $d$ is the de Rham differential. This induces a map 
\begin{equation}
\kappa \maps H^{p}(M,D_{p}^{\bullet}) \to  \cOmega^{p+1}(M),
\end{equation}
where $\cOmega^{p+1}(M)$ are the closed $(p+1)$-forms on $M$. 
If $[\xi] \in H^{p}(M,D_{p}^{\bullet})$, then $\kappa([\xi])$
is called the {\bf $(\mathbf{p+1})$-curvature} of $[\xi]$.
If $ \jmath \maps H^{k}(M,\Z) \to H^{k}(M,\R)$ is the map induced from
the inclusion of the constant sheaves $\Z \hookrightarrow \R$, then one
can prove that $\jmath(c([\xi])) \in H^{p+1}(M,\R)$ corresponds to the class 
$(-1)^{p-1}[\kappa([\xi])] \in H^{p+1}_{\mathrm{DR}}(M)$ via the
isomorphism between \v{C}ech and de Rham cohomology. 

\subsection{Prequantization} \label{symp_prequant}
A symplectic manifold $(M,\omega)$  admits a \textbf{prequantization}
iff the cohomology class $[\omega]$ lies in the image of the map $H^{2}(M,\Z)
\to H^{2}(M,\R)\cong H_{\mathrm{DR}}^{2}(M)$. By virtue of Eq.\ \ref{chern}, there exists a
Deligne class in $H^{1}(M,D_{1}^{\bullet})$ whose 2-curvature is
$\omega$. By definition, a representative of this class defined on a
good cover $\{U_{i}\}$ is a collection of 1-forms $\{ \theta_{i} \in
\Omega^{1}(U_{i}) \}$, and $U(1)$-valued functions $\{g_{ij} \maps
U_{ij} \to U(1)\}$ such that
\begin{align*}
\omega&=d\theta_{i} \quad \text{on $U_{i}$},\\
\theta_{j} -\theta_{i} &= g^{-1}_{ij}dg_{ij} \quad \text{on $U_{ij}$},\\
g_{jk}g^{-1}_{ik}g_{ij}&=1 \quad \text{on $U_{ijk}$}.
\end{align*}
Hence a 1-cocycle is a principal $U(1)$-bundle $P \stackrel{\pi}{\to}
M$ equipped with a connection $\theta \in \Omega^{1}(P)$ with
curvature $d\theta = \pi^{\ast}\omega$. A symplectic manifold equipped
with such a 1-cocycle is said to be \textbf{prequantized}.

The Deligne 1-cocycle also gives, of course, a trivialization of the
2-form $\omega$, and therefore the transitive Lie algebroid
$(A,\tabrac{\cdot}{\cdot},\rho)$ over $M$ equipped with the
splitting $s\maps TM \to A$. However in this case, the
functions $\{g_{ij} \maps U_{ij} \to U(1)\}$ are the transition
functions of the bundle $P$. Therefore, by identifying $\u(1)$ with
$\R$, we see that $A$ is isomorphic to the \textbf{Atiyah
  algebroid} $TP/U(1)$. A point in $A$ corresponds to a
vector field along the fiber $\pi^{-1}(x)$ that is invariant under the
right $U(1)$ action. Hence a global section of $A$
corresponds to a $U(1)$-invariant vector field on $P$.

In general, splittings of $0 \to M \times \R \to A \to TM \to 0$
correspond to connection 1-forms on $P$. 
The connection 1-form $\theta \in \Omega^{1}(P)$ induces a `left-splitting'
$\hat{\theta} \maps A \to M \times \R$ such that $\hat{\theta} \circ s =0$.
It is straightforward to show that $a \in \Gamma(A)^{s}$ if
and only if 
\[
\L_{a} \theta =0.
\]
That is, a section of the Atiyah algebroid preserves the splitting if
and only if it preserves the corresponding connection on $P$. For a
prequantized symplectic manifold, the Lie algebra
$\Gamma(A)^{s}$ is a Lie sub-algebra of derivations on
$\cinf(P)_{\C}$ and therefore on the global sections of the associated
line bundle of $P$. Proposition \ref{lie_alg_iso} then implies that we
have a faithful representation, or quantization, of the Poisson
algebra $\left(\cinf(M), \brac{\cdot}{\cdot} \right)$.

\subsection{The Kostant-Souriau central extension} \label{symp_extend}
If $(M,\omega)$ is a connected symplectic manifold, then we
have a short exact sequence of Lie algebras
\begin{equation} \label{KS1}
0 \to \u(1) \to \cinf(M) \to \Xham \to 0
\end{equation}
The underlying Lie algebra of the Poisson algebra is known as the
Kostant-Souriau central extension of the Lie algebra of Hamiltonian
vector fields \cite{Kostant:1970}. If $\sigma \maps \Xham \to \cinf(M)$ is
a splitting of the underlying sequence of vector spaces, then the
failure of $\sigma$ to be a strict (i.e.\ bracket-preserving) Lie algebra
morphism is measured by the difference
\[
\{\sigma(v_{1}),\sigma(v_{2})\} - \sigma([v_{1},v_{2}])
\]
which represents a degree 2 class in the
Chevalley-Eilenberg cohomology $H^{2}_{\mathrm{CE}}(\Xham,\R)$. This class
can be represented by using the symplectic form. More specifically, pick a point
$x\in M$ and let $c \in \Hom(\Lambda^{2}\Xham,\R)$ be the cochain
given by:
\[
c(v,v')= - \omega(v,v') \vert_{x}, \quad \forall v,v' \in \Xham.
\] 
The fact that $c$ is a cocycle follows from the bracket
$\brac{\cdot}{\cdot}$ satisfying the Jacobi identity.
One can show that the class $[c]$ does not depend on the choice of $x\in
M$.

If $(M,\omega)$ is a prequantized connected symplectic manifold, then
Prop.\ \ref{lie_alg_iso} implies that the `quantized Poisson algebra'
gives an isomorphic central extension
\[
0 \to \u(1) \to \Gamma(A)^{s} \to \Xham \to 0.
\]
This central extension is
responsible for introducing phases into the quantized system. Two
functions $f$ and $f'$ differing by a constant $r \in \u(1)$ will have
the same Hamiltonian vector fields and therefore give the same flows
on $M$. However, their quantizations will give unitary transformations
which differ by a phase $\exp (2\pi \i r )$.

\section{$2$-plectic geometry} \label{2plectic_sec} In this section we
give an overview of $2$-plectic geometry.  Motivation for the
definitions presented here, as well as examples and applications can
be found in previous work \cite{Baez:2008bu,Baez:2009uu}. All of the
following definitions and propositions generalize to arbitrary
$n$-plectic manifolds, so we refer the reader to our recent work
\cite{Rogers:2010nw} for proofs and more details.

\begin{definition}
\label{2-plectic_def}
A $3$-form $\omega$ on a smooth manifold $M$ is 
{\bf 2-plectic}, or more specifically
a {\bf 2-plectic structure}, if it is both closed:
\[
    d\omega=0,
\]
and nondegenerate:
\[
    \forall v \in T_{x}M,\ \iota_{v} \omega =0 \Rightarrow v =0
\]
If $\omega$ is a $2$-plectic form on $M$ we call the pair $(M,\omega)$ 
a {\bf 2-plectic manifold}.
\end{definition}

The $2$-plectic structure induces an injective map from the
space of vector fields on $M$ to the space of 2-forms on $M$. This leads
us to the following definition:

\begin{definition} \label{hamiltonian}
Let $(M,\omega)$ be a $2$-plectic manifold.  A 1-form $\alpha$ on $M$ 
is {\bf Hamiltonian} if there exists a vector field $v_\alpha$ on $M$ such that
\[
d\alpha= -\ip{\alpha} \omega.
\]
We say $v_\alpha$ is the {\bf Hamiltonian vector field} corresponding to $\alpha$. 
The set of Hamiltonian 1-forms and the set of Hamiltonian vector
fields on a $2$-plectic manifold are both vector spaces and are denoted
as $\ham$ and $\Xham$, respectively.
\end{definition}

The Hamiltonian vector field $v_\alpha$ is unique if
it exists, but there may
be 1-forms $\alpha$ having no Hamiltonian vector field.  
Furthermore, two distinct Hamiltonian 1-forms may differ by a closed
1-form and therefore share the same Hamiltonian vector field.

We can generalize the Poisson bracket on functions in symplectic geometry by  
defining a bracket on Hamiltonian 1-forms. 
\begin{definition}
\label{semi-bracket.defn}
Given $\alpha,\beta\in \ham$, the {\bf bracket} $\brac{\alpha}{\beta}$
is the 
\break
1-form given by 
\[  \brac{\alpha}{\beta} = \ip{\beta}\ip{\alpha}\omega .\]
\end{definition}

\begin{prop}\label{semi-bracket} Let $\alpha,\beta,\gamma \in \ham$ and let
$v_\alpha,v_\beta,v_\gamma$ be the respective Hamiltonian
vector fields.  The bracket $\brac{\cdot}{\cdot}$ has the following properties:
  \begin{enumerate}
\item The bracket of Hamiltonian forms is Hamiltonian:
  \begin{eqnarray} \label{semi-bracket.closed}
    d\brac{\alpha}{\beta} = -\iota_{[v_\alpha,v_\beta]} \omega, 
  \end{eqnarray}
so in particular we have 
\[     v_{\brac{\alpha}{\beta}} = [v_\alpha,v_\beta]  .\]
\item The bracket is skew-symmetric: 
  \begin{eqnarray}
    \brac{\alpha}{\beta} = -\brac{\beta}{\alpha}
  \end{eqnarray}

\item The bracket satisfies the Jacobi identity up to an exact 1-form:
\begin{eqnarray}
    \brac{\alpha}{\brac{\beta}{\gamma}} -
    \brac{\brac{\alpha}{\beta}}{\gamma} 
    -\brac{\beta}{\brac{\alpha}{\gamma}}=d \ip{\alpha}\ip{\beta}\ip{\gamma}\omega.
  \end{eqnarray}
\end{enumerate}
\end{prop}
\begin{proof}
See Propositions 3.5 and 3.6 in \cite{Rogers:2010nw}.
\end{proof}
Note that Eq.\ \ref{semi-bracket.closed} in the above proposition
implies that $\Xham$ is a Lie algebra.

\section{Courant algebroids} \label{courant_sec}
Here we recall some basic facts and examples of Courant
algebroids and then we proceed to describe \v{S}evera's classification of exact Courant
algebroids. There are several equivalent definitions of a Courant
algebroid found in the literature. The following
definition, due to Roytenberg \cite{Roytenberg_thesis}, is equivalent
to the original definition given by Liu, Weinstein, and Xu \cite{Liu:1997}.

\begin{definition} \label{courant_algebroid} A {\bf Courant algebroid}
  is a vector bundle $C \to M$ equipped with a nondegenerate symmetric
  bilinear form $\innerprod{\cdot}{\cdot}$ on the bundle, a
  skew-symmetric bracket $\cbrac{\cdot}{\cdot}$ on $\Gamma(C)$, and a
  bundle map (called the {\bf anchor}) $\rho \maps C \to TM$ such that
  for all $ e_{1},e_{2},e_{3}\in \Gamma (C)$ and for all $f,g \in
  \cinf(M)$ the following properties hold:
\begin{enumerate}
\item{$\cbrac{e_1}{\cbrac{e_2}{e_3}} - \cbrac{\cbrac{e_1}{e_2}}{e_3} 
-\cbrac{e_2}{\cbrac{e_1}{e_3}}=-DT(e_{1},e_{2},e_{3}),$  
}

\item{$\rho(\cbrac{e_{1}}{e_{2}})=[\rho(e_{1}),\rho(e_{2})]$, \label{axiom2}
} 

\item{$\cbrac{e_{1}}{fe_{2}}=f\cbrac{e_{1}}{e_{2}}+\rho (e_{1})(f)e_{2}-\half
\langle e_{1},e_{2}\rangle {D}f$,
}

\item{$\langle {D}f,{D}g\rangle =0$,}

\item{$\rho(e_{1})\left(\langle
e_{2},e_{3}\rangle\right) =\langle \cbrac{e_{1}}{e_{2}}+\half {D}\langle
e_{1},e_{2}\rangle ,e_{3}\rangle +\langle e_{2},\cbrac{e_{1}}{e_{3}}+\half
{D}\langle e_{1},e_{3}\rangle \rangle$,}
\end{enumerate}
where $[\cdot,\cdot]$ is the Lie bracket of vector fields,
$D \maps \cinf(M) \to \Gamma(C)$ is the map defined by $\innerprod{D f}{e}=\rho(e)f$, and 
\[
T(e_{1},e_{2},e_{3})= \frac{1}{6}
\left(\innerprod{\cbrac{e_1}{e_2}}{e_3} +
\innerprod{\cbrac{e_3}{e_1}}{e_2} + \innerprod{\cbrac{e_2}{e_3}}{e_1} \right). 
\]
\end{definition}
The bracket in Definition \ref{courant_algebroid} is skew-symmetric,
but the first property implies that it needs only to satisfy the
Jacobi identity ``up to $D T$''. Note that the vector bundle $C \to M$ may
be identified with $C^{\ast} \to M$ via the bilinear form
$\innerprod{\cdot}{\cdot}$ and therefore we have the dual map
\[ \rho^{\ast} \maps T^{\ast}M \to C.\] 
Hence the map $D$ is simply the pullback of the de Rham differential by $\rho^{\ast}$. 

There is an alternate definition given by \v{S}evera \cite{Severa1}
for Courant algebroids which uses a bracket operation on sections
that satisfies a Jacobi identity but is not skew-symmetric.
\begin{definition}\label{alt_courant_algebroid}
  A {\bf Courant algebroid} is a vector bundle $ C\rightarrow M $
  together with a nondegenerate symmetric bilinear form $ \innerprod
  {\cdot }{\cdot } $ on the bundle, a bilinear operation $\dbrac{\cdot}{\cdot}$ on
  $ \Gamma (C) $, and a bundle map $ \rho \maps C\rightarrow TM $ such
  that for all $ e_{1},e_{2},e_{3}\in \Gamma (C)$ and for all $f \in
  \cinf(M)$ the following properties hold:
\begin{enumerate}
\item{$
    \dbrac{e_{1}}{\dbrac{e_{2}}{e_{3}}}=\dbrac{\dbrac{e_{1}}{e_{2}}}{e_{3}}+\dbrac{e_{2}}{\dbrac{e_{1}}{e_{3}}}$,}
\item {$\rho (\dbrac{e_{1}}{e_{2}})=[\rho (e_{1}),\rho (e_{2})]$,} 
\item {$ \dbrac{e_{1}}{fe_{2}}=f\dbrac{e_{1}}{e_{2}}+\rho (e_{1})(f)e_{2}$,}
\item {$ \dbrac{e_{1}}{e_{1}}=\half D \innerprod {e_{1}}{e_{1}}$,}
\item {$ \rho (e_{1})\left(\innerprod {e_{2}}{e_{3}}\right)=\innerprod
    {\dbrac{e_{1}}{e_{2}}}{e_{3}}+\innerprod {e_{2}}{\dbrac{e_{1}}{e_{3}}}$,}
\end{enumerate}
where $[\cdot,\cdot]$ is the Lie bracket of vector fields, and
$D \maps \cinf(M) \to \Gamma(C)$ is the map defined by $\innerprod{D f}{e}=\rho(e)f$.
\end{definition}
Roytenberg \cite{Roytenberg_thesis} showed that $C \to M$ is a Courant
algebroid in the sense of Definition \ref{courant_algebroid} with
bracket $\cbrac{\cdot}{\cdot}$, bilinear form
$\innerprod{\cdot}{\cdot}$ and anchor $\rho$ if and only if $C \to M$
is a Courant algebroid in the sense of Definition
\ref{alt_courant_algebroid} with the same anchor and bilinear form but
with bracket $\dbrac{\cdot}{\cdot}$ given by
\begin{equation}
 \dbrac{e_{1}}{e_{2}} = \cbrac{e_{1}}{e_{2}} + \half D
 \innerprod{e_{1}}{e_{2}}. \label{dorfman}
\end{equation}
All Courant algebroids mentioned in this paper are Courant algebroids in
the sense of Definition \ref{courant_algebroid}. We introduced
Definition \ref{alt_courant_algebroid} mainly to connect our
discussion here with previous results in the literature.

\begin{example} \label{standard}
An important example of a Courant algebroid is the \textbf{standard
  Courant algebroid} $C=TM \oplus T^{\ast}M$ over any manifold
  $M$ equipped with the \textbf{standard Courant
  bracket}:
\begin{equation} 
\cbrac{v_{1} + \alpha_{1}}{v_{2} + \alpha_{2}}=
[v_{1},v_{2}] +\L_{v_{1}}\alpha_{2}- \L_{v_{2}}\alpha_{1} -
\half d\innerprodm{v_{1} + \alpha_{1}}{v_{2} + \alpha_{2}}, \label{standard_bracket}
 \end{equation}
where
\begin{equation} \label{skew}
\innerprodm{v_{1} + \alpha_{1}}{v_{2} + \alpha_{2}} = \ip{1}\alpha_{2}
- \ip{2}\alpha_{1}
\end{equation}
is the \textbf{standard skew-symmetric pairing}.
The bilinear form is given by the \textbf{standard symmetric pairing}:
\begin{equation}
\innerprodp{v_{1} + \alpha_{1}}{v_{2} + \alpha_{2}}=
\ip{1}\alpha_{2} + \ip{2}\alpha_{1}. \label{standard_innerprod}
\end{equation}
The anchor $\rho \maps C \to TM$ is the projection
map, and $D=d$ is the de Rham differential. The bracket
$\cbrac{\cdot}{\cdot}$ is the skew-symmetrization of the \textbf{standard
  Dorfman bracket} \cite{Dorfman1,Dorfman2}:
\begin{equation} \label{std_dorf}
\dbrac{v_{1} + \alpha_{1}}{v_{2} + \alpha_{2}}=
[v_{1},v_{2}] + \L_{v_{1}}\alpha_{2} - \ip{2}d\alpha_{1},
\end{equation}
which plays the role of the bracket given in
Definition \ref{alt_courant_algebroid}. 
\end{example}

The standard Courant algebroid is the prototypical example of
an \textbf{exact Courant algebroid} \cite{Bressler-Chervov}.
\begin{definition} \label{exact} A Courant algebroid $C \to M$ with
  anchor map $\rho \maps C \to TM$ is {\bf exact} iff
\[
0 \to T^{\ast}M \stackrel{\rho^{\ast}}{\to} C \stackrel{\rho}{\to} TM
\to 0\]
is an exact sequence of vector bundles.
\end{definition}

\subsection{The \v{S}evera class of an exact Courant algebroid} \label{class}
\v{S}evera's classification \cite{Severa1} originates in the idea that a particular
kind of splitting of the above short exact sequence corresponds to defining a
connection.
\begin{definition}\label{connection}
A {\bf splitting} of an exact Courant algebroid $C$ over a manifold $M$ is a map of
vector bundles $s \maps TM \to C $ such that
\begin{enumerate}
\item{ $\rho \circ s = \id_{TM}$,}
\item{$\innerprod{s(v_{1})}{s(v_{2})}=0$ for all $v_{1},v_{2} \in TM$,}
\end{enumerate}
where $\rho \maps C \to TM$ and $\innerprod{\cdot}{\cdot}$ are the
anchor and bilinear form, respectively.
\end{definition}
In other words, a splitting of an exact Courant algebroid is an isotropic splitting of
the sequence of vector bundles. Bressler and Chervov call splittings `connections'
\cite{Bressler-Chervov}.
If $s$ is a splitting and $B \in
\Omega^2(M)$ is a 2-form then one can construct a new splitting:
\begin{equation}
\left(s+B\right)(v)=s(v)+ \rho^{\ast}B(v,\cdot). \label{2-form_action}
\end{equation}
Furthermore, one can show that any two splittings on an exact Courant
algebroid must differ by a 2-form on
$M$ in this way. Hence the space of splittings on an exact Courant algebroid
is an affine space modeled on the vector space of 2-forms
$\Omega^{2}(M)$ \cite{Bressler-Chervov}.

The failure of a splitting to preserve the bracket gives a suitable
notion of curvature:
\begin{definition}[\cite{Bressler-Chervov}]
  If $C$ is an exact Courant algebroid over $M$ with bracket
  $\cbrac{\cdot}{\cdot}$ and $s \maps TM \to C$ is a splitting then
  the {\bf curvature} is a map $F \maps TM \times TM \to C$ defined
  by
\[ F(v_{1},v_{2}) = \cbrac{s(v_{1})}{s(v_{2})} - s\left( \left[
    v_{1},v_{2} \right] \right).
\]
\end{definition}
If $F$ is the curvature of a splitting $s$, then given $v_{1},v_{2}
\in TM$, it follows from exactness and axiom \ref{axiom2} in Definition
\ref{courant_algebroid} that there exists a 1-form
$\alpha_{v_{1},v_{2}} \in \Omega^1(M)$ such that
$F(v_{1},v_{2})=\rho^{\ast}(\alpha_{v_{1},v_{2}})$. Since $s$ is a
splitting, its image is isotropic in $C$. Therefore for any $v_{3}
\in TM$ we have:
\[
\innerprod{F(v_{1},v_{2})}{s(v_{3})} = 
\innerprod{\cbrac{s\left(v_{1}\right)}{s\left(v_{2}\right)}}{s(v_{3})}.
\] 
The above formula allows one to associate the curvature $F$ to a 3-form on $M$:  
\begin{prop} \label{class_eq}
  Let $C$ be an exact Courant algebroid over a manifold $M$ with
  bracket $\cbrac{\cdot}{\cdot}$ and bilinear form
  $\innerprod{\cdot}{\cdot}$. Let $s \maps TM \to C$ be a splitting on $C$. Then
  given vector fields $v_{1}, v_{2}, v_{3}$ on $M$:
\begin{enumerate}
\item{ 
The function
\[
\omega(v_{1},v_{2},v_{3})= \innerprod{\cbrac{s\left(v_{1}\right)}{s\left(v_{2}\right)}}{s(v_{3})}
\]
defines a closed 3-form on $M$.}
\item{If $\theta \in \Omega^{2}(M)$ is a 2-form and
    $\tilde{s}=s+\theta$ then
\begin{align*}
\tilde{\omega}(v_{1},v_{2},v_{3}) &=
\innerprod{\left [
  \tilde{s}\left(v_{1}\right),\tilde{s}\left(v_{2}\right) \right ]_{C}}{\tilde{s}(v_{3})}\\
&=\omega(v_{1},v_{2},v_{3}) + d\theta(v_{1},v_{2},v_{3}).
\end{align*}
}
\end{enumerate}
\end{prop}
\begin{proof}
  The statements in the proposition are proved in Lemmas 4.2.6, 4.2.7,
  and 4.3.4 in the paper by Bressler and Chervov
  \cite{Bressler-Chervov}.  In their work they define a Courant
  algebroid using Definition \ref{alt_courant_algebroid}, and therefore
  their bracket satisfies the Jacobi identity, but is not
  skew-symmetric.  In our notation, their definition of the curvature
  3-form is:
\[
\nu(v_{1},v_{2},v_{3})= \innerprod{\dbrac{s\left(v_{1}\right)}{s\left(v_{2}\right)}}{s(v_{3})}.
\] 
In particular they show that $\dbrac{\cdot}{\cdot}$
satisfying the Jacobi identity implies $\nu$ is
closed. The bracket $\cbrac{\cdot}{\cdot}$ does not satisfy the Jacobi
identity in general. However the isotropicity of the splitting and
Eq.\ \ref{dorfman} imply
\[
\dbrac{s(v_1)}{s(v_2)}  = \cbrac{s(v_1)}{s(v_2)} \quad \forall
   v_1, v_2 \in TM.
\]
Hence $\nu=\omega$, so all the needed results in
\cite{Bressler-Chervov} apply here.
\end{proof}

Thus the above  proposition implies that the curvature 3-form of an exact
Courant algebroid over $M$ gives a well-defined cohomology class in
$H^{3}_{\mathrm{DR}}(M)$, independent of the choice of
splitting. 
\begin{definition}[\cite{Gualtieri:2007}] \label{Severa_class} The {\bf \v{S}evera class}
  of an exact Courant algebroid with bracket $\cbrac{\cdot}{\cdot}$ and bilinear
  form $\innerprod{\cdot}{\cdot}$ is the cohomology
  class $[-\omega] \in H^{3}_{\mathrm{DR}}(M)$, where
\[
\omega(v_{1},v_{2},v_{3})=
\innerprod{\cbrac{s\left(v_{1}\right)}{s\left(v_{2}\right)}}{s(v_{3})}.
\]
\end{definition}

\section{The Courant algebroid associated to a $2$-plectic
  manifold} \label{geometric} In this section we recall how to
explicitly construct an exact Courant algebroid with \v{S}evera class
$[\omega]$. This is the 3-form analogue of the construction that gives
a transitive Lie algebroid over a pre-symplectic manifold, which was
previously discussed in Sec. \ref{symp_algebroid}.  The approach
given here is essentially identical to those given by Gualtieri
\cite{Gualtieri:2007}, Hitchin \cite{Hitchin:2004ut}, and \v{S}evera
\cite{Severa1} .

Let $(M,\omega)$ be a manifold equipped with a closed 3-form.
A trivialization of $\omega$
is an open cover$\{U_{i}\}$ of $M$ equipped with 2-forms  
$B_{i} \in \Omega^{2}(U_{i})$, and 1-forms $A_{ij} \in \Omega^{1}(U_{ij})$ 
on intersections such that
\begin{equation} \label{cocycle2}
\begin{split}
 \omega \vert_{U_{i}} &= dB_{i} \\
(B_{j}-B_{i}) \vert_{U_{ij}} &= dA_{ij}.
\end{split}
\end{equation}
Given such a trivialization, over each open set $U_{i}$ consider the
bundle $C_{i}=TU_{i} \oplus
T^{\ast}U_{i}\to U_{i}$ equipped with the standard pairing
\begin{equation} \label{twist_innerprod}
\innerprodp{v_{1} + \alpha_{1}}{v_{2} + \alpha_{2}}_{i}= \ip{1}\alpha_{2}
+ \ip{2}\alpha_{1},
\end{equation}
$v_{1},v_{2} \in \X(U_{i}),$  $\alpha_{1},\alpha_{2} \in\Omega^{1}(U_{i}),$
which has signature $(n,n)$. On double intersections, it's
easy to see that 
\[
\innerprodp{v_{1} + \ip{1}dA_{ij} +\alpha_{1}}{v_{2} + \ip{2}dA_{ij} +
  \alpha_{2}}_{i} = \innerprodp{v_{1} + \alpha_{1}}{v_{2} + \alpha_{2}}_{i}.
\]
Hence the 2-forms $\{dA_{ij}\}$ generate transition functions
\begin{align*}
G_{ij} \maps U_{ij} \to SO(n,n),\\
G_{ij}(x)= 
\begin{pmatrix}
1 & 0\\
dA_{ij} \vert_{x} & 1
\end{pmatrix},
\end{align*}
which satisfy the cocycle conditions on $U_{ijk}$ by virtue of Eq.\
\ref{cocycle2}. Therefore, we have over  $M$ the vector bundle
\[
C= \coprod_{x \in M} T_{x}U_{i} \oplus T^{\ast}_{x}U_{i} / \sim,
\]
equipped with a bilinear form denoted as $\innerprodp{\cdot}{\cdot}$.
$C$ sits in the exact sequence
\[
0 \to T^{\ast}M \stackrel{\jmath}{\rightarrow} C \stackrel{\rho}\to TM \to 0,
\]
where the anchor $\rho$ is induced by the projection $T^{\ast}U_{i} \oplus TU_{i}
\to TU_{i}$, and $\jmath$ is the inclusion.

The 2-forms $\{B_{i}\}$ induce a bundle map  $s \maps TM \to C$
\begin{equation} \label{canonical_split}
s(v_{x}) = v_{x} - \iota_{v_{x} }B_{i} \quad \text{if $x\in U_{i}$},
\end{equation}
It follows from Eq.\ \ref{cocycle2} that $s$ is well-defined when $x \in
U_{ij}$. It is easy to see that this map is an isotropic splitting
(Def.\ \ref{connection}).
Hence every section $e\in \Gamma(C)$ can be uniquely expressed
as
\[
e= s(v) + \alpha,
\]
for some $v\ \in \X(M)$ and $\alpha \in \Omega^{1}(M)$.  As
before, we use $s$ to also denote the map $\Gamma(s) \maps \X(M) \to
\Gamma(C)$. The anchor map is just
\begin{equation} \label{twist_anchor}
\rho \bigl(s(v)+\alpha \bigr)=v.
\end{equation}

Given sections $s(v_{1})+\alpha_{1}$, $s(v_{2}) + \alpha_{2} \in
\Gamma(C)$, a local calculation using Eq.\
\ref{canonical_split}  gives
\begin{equation} \label{split_innerprod}
\begin{split}
\innerprodp{s(v_{1})+\alpha_{1}}{s(v_{2})+\alpha_{2}}
&= \ip{1}\alpha_{2} - \ip{1}\ip{2}B_{i} + \ip{2}\alpha_{1} -
\ip{2}\ip{1}B_{i}\\
&= \innerprodp{v_{1}+\alpha_{1}}{v_{2}+\alpha_{2}}.
\end{split}
\end{equation}
The above equality holds, in fact, for any splitting $s' \maps TM \to
C$, since $s-s'$ is a 2-form on $M$ and therefore skew-symmetric.
The bracket on $\Gamma(C)$ is defined over the open set
$U_{i}$ by:
\[
\tcbrac{s(v_{1}) + \alpha_{1}}{s(v_{2}) +
  \alpha_{2}}\vert_{U_{i}} =\sbrac{s(v_{1}) + \alpha_{1}}{s(v_{2}) +
  \alpha_{2}}_{i}
\]
where $\sbrac{\cdot}{\cdot}_{i}$ is the standard Courant bracket
(\ref{standard_bracket}) on $C_{i}$. Since the 2-forms $\{dA_{ij}\}$
are closed, it follows by direct computation that on double
intersections $U_{ij}$:
\[
\sbrac{G_{ij} (v_{1} + \alpha_{1})}{G_{ij}(v_{2} +
  \alpha_{2})}_{i} = G_{ij} \bigl(\sbrac{v_{1} + \alpha_{1}}{v_{2} +
  \alpha_{2}}_{i} \bigr).
\]
Hence the bracket $\tcbrac{\cdot}{\cdot}$ is indeed globally well-defined. 
Using the local definition of the bracket and the splitting, as well
as the fact that $dB_{i}=\omega$, it is easy to show that
\begin{equation} \label{twist_bracket}
\begin{split}
\tcbrac{s(v_{1}) + \alpha_{1}}{s(v_{2}) +
  \alpha_{2}} &= s \bigl ([v_{1},v_{2}] \bigr) +
\L_{v_{1}}\alpha_{2}- \L_{v_{2}}\alpha_{1} \\
 & \quad -\half d\innerprodm{v_{1} + \alpha_{1}}{v_{2} + \alpha_{2}}
-\ip{2}\ip{1}\omega.
\end{split}
\end{equation}
The bracket $\tcbrac{\cdot}{\cdot}$ is called the \textbf{twisted Courant bracket}. 
A analogous construction using the standard Dorfman bracket (\ref{std_dorf}) 
on $C_{i}$ gives the \textbf{twisted Dorfman
  bracket}:
\begin{equation} \label{twist_bracket_dorf}
\tdbrac{s(v_{1}) + \alpha_{1}}{s(v_{2}) + \alpha_{2}}=
s\bigl([v_{1},v_{2}]\bigr) + \L_{v_{1}}\alpha_{2} - \ip{2}d\alpha_{1} -\ip{2}\ip{1}\omega.
\end{equation}
These brackets were studied in detail by \v{S}evera and Weinstein
\cite{Severa1,Severa-Weinstein}.

It is straightforward to check that $C \to M$ equipped with the
aforementioned bilinear form, anchor, and bracket
$\tcbrac{\cdot}{\cdot}$ is an exact Courant algebroid (Definition
\ref{courant_algebroid}). Just as in Lie algebroid case (Sec.\
\ref{symp_algebroid}), the construction of $C$
is independent of the choice of trivialization up to a splitting-preserving
isomorphism.

A direct calculation shows that
\[
-\omega(v_{1},v_{2},v_{3}) =
\innerprodp{\tcbrac{s\left(v_{1}\right)}{s\left(v_{2}\right)}}{s(v_{3})}.
\]
Hence by Prop.\ \ref{class_eq}, the Courant algebroid $C$ has
\v{S}evera class $[\omega]$. Of course, we are interested in the
special case when $\omega$ is a 2-plectic structure.
We summarize the above discussion with the following proposition:
\begin{prop} \label{2plectic_courant}
Let $(M,\omega)$ be a $2$-plectic manifold. Up to isomorphism, there
exists a unique exact Courant
algebroid $C$ over $M$, with
bilinear form $\innerprodp{\cdot}{\cdot}$,  
anchor map $\rho$, and 
bracket $\tcbrac{\cdot}{\cdot}$ 
given in Eqs.\ \ref{twist_innerprod}, \ref{twist_anchor}, and
\ref{twist_bracket}, respectively,
and  equipped with a splitting whose curvature is $-\omega$.
\end{prop}

\section{Lie $2$-algebras} \label{L2A_sec} Both the Courant bracket
and the bracket on Hamiltonian 1-forms are, roughly, Lie brackets
which satisfy the Jacobi identity up to an exact 1-form. This leads us
to the notion of a Lie $2$-algebra.  In general, a Lie $n$-algebra is
a $n$-term $L_{\infty}$-algebra.  It consists of a graded vector space
concentrated in degrees $0,\ldots,n-1$ and is equipped with a
collection of skew-symmetric $k$-ary brackets, for $1 \leq k \leq
n+1$, that satisfy a generalized Jacobi identity
\cite{Lada-Markl,LS}. In particular, the $k=2$ bilinear bracket
behaves like a Lie bracket that only satisfies the ordinary Jacobi
identity up to higher coherent chain homotopy.  Baez and Crans showed
that Lie 2-algebras are equivalent to categories internal to the
category of vector spaces over $\R$ equipped with structures analogous
to those of a Lie algebra, for which the usual law involving the
Jacobi identity holds only up to natural isomorphism \cite{HDA6}.
(Note that what we call a Lie 2-algebra is called a `semistrict Lie
2-algebra' in \cite{HDA6}, \cite{Baez:2008bu}, and
\cite{Roytenberg_L2A}.)

As $L_{\infty}$-algebras, Lie 2-algebras are relatively easy to work
with and one can write out the axioms explicitly. Therefore we use the
following definition which is equivalent the usual definition for an
$L_{\infty}$-algebra \cite{LS} when the underlying complex is
concentrated in degrees 0 and 1.
\begin{definition} \label{L2A}
A {\bf Lie 2-algebra} is a $2$-term chain complex of vector spaces
$L_{\bullet} = (L_1\stackrel{d}\rightarrow L_0)$ equipped with:
\begin{itemize}
\item a skew-symmetric chain map $\blankbrac\maps L_{\bullet} \tensor
  L_{\bullet}\to L_{\bullet}$ called the {\bf bracket};
\item a skew-symmetric chain homotopy $J \maps L_{\bullet} \tensor L_{\bullet} \tensor L_{\bullet}
  \to L_{\bullet}$ 
from the chain map
\[     \begin{array}{ccl}  
     L_{\bullet} \tensor L_{\bullet} \tensor L_{\bullet} & \to & L_{\bullet}   \\
     x \tensor y \tensor z & \longmapsto & [x,[y,z]],  
  \end{array}
\]
to the chain map
\[     \begin{array}{ccl}  
     L_{\bullet} \tensor L_{\bullet} \tensor L_{\bullet}& \to & L_{\bullet}   \\
     x \tensor y \tensor z & \longmapsto & [[x,y],z] + [y,[x,z]]  
  \end{array}
\]
called the {\bf Jacobiator},
\end{itemize}
such that the following equation holds:
\begin{equation} \label{big_J}
\begin{array}{c}
  [x,J(y,z,w)] + J(x,[y,z],w) +
  J(x,z,[y,w]) + [J(x,y,z),w] \\ + [z,J(x,y,w)] 
  = J(x,y,[z,w]) + J([x,y],z,w) \\ + [y,J(x,z,w)] + J(y,[x,z],w) + J(y,z,[x,w]).
\end{array}
\end{equation}
\end{definition}

We will also need a suitable notion of morphism: 
\begin{definition}
\label{homo}
Given semistrict Lie $2$-algebras $L=(L_{\bullet},\blankbrac,J)$ and
$L'=(L_{\bullet}',{\blankbrac}^{\prime},J')$ a {\bf morphism} from
$L$ to $L'$ consists of:
\begin{itemize}
\item{a chain map $\phi_{\bullet} \maps L_{\bullet} \to L_{\bullet}'$, and}
\item{a chain homotopy $\Phi \maps L_{\bullet} \tensor L_{\bullet} \to L_{\bullet}'$ from the chain
  map
\[     \begin{array}{ccl}  
     L_{\bullet} \tensor L_{\bullet} & \to & L_{\bullet}'   \\
     x \tensor y & \longmapsto & \phi_{\bullet} \left( [x,y] \right)
  \end{array}
\]
to the chain map
\[     \begin{array}{ccl}  
     L_{\bullet} \tensor L_{\bullet} & \to & L_{\bullet}'   \\
     x \tensor y & \longmapsto & \left [ \phi_{\bullet}(x),\phi_{\bullet}(y) \right]^{\prime},
  \end{array}
\]
}
\end{itemize}
such that the following equation holds:
\begin{equation} \label{coherence}
\begin{array}{l}
\phi_1(J(x,y,z))- J^{\prime}(\phi_0(x),\phi_0(y), \phi_0(z)) = \\
\Phi(x,[y,z]) -\Phi([x,y],z) - \Phi(y,[x,z]) - [\Phi(x,y),\phi_0(z)]^{\prime}\\
+ [\phi_0(x), \Phi(y,z)]^{\prime}- [\phi_0(y),\Phi(x,z)]^{\prime}.
\end{array}
\end{equation}
We say a morphism is {\bf strict} iff $\Phi=0$.
\end{definition}
This definition is equivalent to the definition of a morphism between
$2$-term $L_{\infty}$-algebras \cite{Lada-Markl}.
\begin{definition}
A Lie 2-algebra morphism $(\phi_{\bullet},\Phi) \maps L \to L'$
is a  {\bf quasi-isomorphism} iff the chain map $\phi_{\bullet}$
induces an isomorphism on the homology of the underlying chain
complexes of $L$ and $L'$.
\end{definition}

\subsection{The Lie $2$-algebra from a $2$-plectic manifold}
Any $n$-plectic manifold gives a Lie $n$-algebra which can be
understood as the $n$-plectic analogue of the Poisson algebra
\cite{Rogers:2010nw}. We now review this construction for the
2-plectic case. The underlying $2$-term chain complex of our Lie
2-algebra is:
\[
L_{\bullet} \quad = \quad 
\cinf(M)  \stackrel{d}{\rightarrow} \ham  
\]
where $d$ is the de Rham differential.
This chain complex is well-defined, since 
any exact form is Hamiltonian, with $0$ as its Hamiltonian vector
field. We can construct a chain map 
\[      [\cdot,\cdot] \maps L_{\bullet} \otimes L_{\bullet} \to L_{\bullet} ,\]
by extending the bracket $\brac{\cdot}{\cdot}$ on $\ham$ trivially to $L_{\bullet}$.
In other words, in degree $0$, the chain map is given as in 
Definition \ \ref{semi-bracket.defn}: 
\[  [\alpha,\beta]=\brac{\alpha}{\beta}= \ip{\beta}\ip{\alpha}\omega, \]
and in degrees $1$ and $2$, we set it equal to zero:
\[  [\alpha,f] = 0, \qquad [f,\alpha] = 0, \qquad
    [f,g] = 0.    \]
The precise construction of this Lie $2$-algebra is 
given in the following theorem:
\begin{theorem}
\label{semistrict}
If $(M,\omega)$ is a $2$-plectic manifold, there is a 
 Lie $2$-algebra $L_{\infty}(M,\omega)=(L_{\bullet},\blankbrac,J)$ where:
\begin{itemize}
\item $L_{0} =\ham$,
\item $L_{1}=\cinf(M)$,
\item the differential $L_{1} \stackrel{d}{\to} L_{0}$ is the de Rham differential,
\item the bracket $\blankbrac$ is $\brac{\cdot}{\cdot}$ in degree 0
  and trivial otherwise,
\item the Jacobiator is given by the linear map $J\maps \ham \tensor \ham \tensor 
\ham \to \cinf$, where $J(\alpha,\beta,\gamma) = \ip{\alpha}\ip{\beta}\ip{\gamma}\omega$.
\end{itemize}
\end{theorem}
\begin{proof}
See Theorem 5.2 in \cite{Rogers:2010nw}.
\end{proof}
\subsection{The Lie $2$-algebra from a Courant algebroid}
Similarly, given any Courant algebroid $C \rightarrow M$ with bilinear form
$\innerprod{\cdot}{\cdot}$, bracket $\cbrac{\cdot}{\cdot}$, and anchor
$\rho \maps C \to TM$, one can construct a
$2$-term chain complex 
\[
L_{\bullet} \quad = \quad \cinf(M) \stackrel{D}{\rightarrow} \Gamma(C),
\]
with differential $D=\rho^{\ast}d$ where $d$ is the de Rham differential. The bracket
$\cbrac{\cdot}{\cdot}$ on global sections can be extended to a chain
map $\sbrac{\cdot}{\cdot} \maps L_{\bullet} \tensor L_{\bullet} \to L_{\bullet}$. If $e_1,e_2$ are
degree 0 chains then $\sbrac{e_{1}}{e_{2}}$ is the original bracket.
If $e$ is a degree 0 chain and $f,g$ are degree 1 chains, then we
define:
\begin{align*}
\sbrac{e}{f} &= -\sbrac{f}{e} = \half \innerprod{e}{D f}  \\
\sbrac{f}{g}&=0.
\end{align*}   
It was shown by Roytenberg and Weinstein \cite{Roytenberg-Weinstein}
that this extended bracket gives a $L_{\infty}$-algebra. Roytenberg's
later work \cite{Roytenberg_graded,Roytenberg_L2A} implies that a
brutal truncation of this $L_{\infty}$-algebra is a Lie 2-algebra
whose underlying complex is $L_{\bullet}$.  For the Courant algebroid
$C$ constructed in Section \ref{geometric}, their result implies:
\begin{theorem}\label{courant_L2A}
If $C$ is the exact Courant algebroid given in Proposition \ref{2plectic_courant}
then there is a  Lie $2$-algebra  
$L_{\infty}(C)=(L_{\bullet},\blankbrac, J)$ where:
\begin{itemize}
\item $L_{0}=\Gamma(C)$,
\item $L_{1}=\cinf(M)$,
\item the differential $L_{1} \stackrel{D}{\to} L_{0}$ is $D=\rho^{\ast}d$,
\item{the bracket $\sbrac{\cdot}{\cdot}$ is 
\[
[e_{1},e_{2}]= \tcbrac{e_{1}}{e_{2}} \quad  \text{in degree 0}
\]
and
\[
  [e,f]=-[f,e]=\half \innerprodp{e}{df} \quad \text{in degree 1},
\]
}
\item the Jacobiator is the linear map $J\maps \Gamma(C) \tensor \Gamma(C) \tensor \Gamma(C) \to \cinf(M)$ defined by
\begin{align*} 
J(e_{1},e_{2},e_{3}) &= -T(e_{1},e_{2},e_{3})\\
&=-\frac{1}{6}
\left(\innerprodp{\tcbrac{e_1}{e_2}}{e_3} +
\innerprodp{\tcbrac{e_3}{e_1}}{e_2} \right. \\
 & \left.\quad + \innerprodp{\tcbrac{e_2}{e_3}}{e_1} \right). 
\end{align*}
\end{itemize}
\end{theorem}
More precisely, the theorem follows from Example 5.4 of
\cite{Roytenberg_L2A} and Section 4 of \cite{Roytenberg_graded}. On
the other hand, the original construction of Roytenberg and Weinstein
gives a $L_{\infty}$-algebra on the complex:
\[
0 \rightarrow \ker D \stackrel{\iota}{\rightarrow} \cinf(M)
\stackrel{D}{\rightarrow}\Gamma(C),
\]
with trivial structure maps $l_{n}$ for $n \geq 3$. Moreover, the map
$l_{2}$ (corresponding to the bracket $\sbrac{\cdot}{\cdot}$ given above) is trivial in
degree $>1$ and the map $l_{3}$ (corresponding to the Jacobiator $J$) is
trivial in degree $>0$. Hence these maps induce the above Lie 2-algebra
structure on $\cinf(M)\stackrel{D}{\rightarrow}\Gamma(C)$.

\section{The algebraic relationship between 2-plectic and
  Courant} \label{algebraic} In Section \ref{geometric}, we described
how one can construct over a 2-plectic manifold $(M,\omega)$, an exact
Courant algebroid
$(C,\tcbrac{\cdot}{\cdot},\innerprodp{\cdot}{\cdot},\rho)$ equipped
with a splitting $s \maps TM \to C$ whose curvature is $-\omega$. In
this section, we show there is a complementary algebraic
relationship. We can interpret these results as the 2-plectic
analogues of those given in Section \ref{symp_preserve}.
\begin{theorem} \label{main_thm} Let $(M,\omega)$ be a $2$-plectic
  manifold and let $C$ be its corresponding Courant
  algebroid. Let $L_{\infty}(M,\omega)$ and $L_{\infty}(C)$ be the 
  Lie 2-algebras corresponding to $(M,\omega)$ and $C$,
  respectively.  There exists a morphism of Lie 2-algebras embedding
  $L_{\infty}(M,\omega)$ into $L_{\infty}(C)$.
\end{theorem} 

Before we prove the theorem, we introduce some technical
lemmas to ease the calculations. Recall from Eq.\ \ref{skew} that the
formula for the standard skew-symmetric pairing
on $\X(M) \oplus \Omega^{1}(M)$:
\[
\innerprodm{v_{1} + \alpha_{1}}{v_{2} + \alpha_{2}} =
\ip{1}\alpha_{2} - \ip{2}\alpha_{1}.
\]
In what follows, by the symbol ``$\cp$'' we mean cyclic permutations of the symbols $\alpha,\beta,\gamma$.
\begin{lemma}\label{calc1}
If $\alpha, \beta \in \ham$ with corresponding Hamiltonian vector fields
$v_{\alpha},v_{\beta}$, then
$\L_{v_{\alpha}}\beta=\brac{\alpha}{\beta} + d \ip{\alpha}\beta$.
\end{lemma}
\begin{proof} Since $\L_{v} = \iota_v d + d \iota_v$,
\[  \L_{v_{\alpha}}{\beta} = 
\ip{\alpha} d \beta + d \ip{\alpha} \beta =
-\ip{\alpha}\ip{\beta} \omega + d \ip{\alpha} \beta =
\brac{\alpha}{\beta} + d \ip{\alpha} \beta .\]
\end{proof}

\begin{lemma} \label{calc2}
If $\alpha,\beta,\gamma \in \ham$ with corresponding Hamiltonian vector fields
$v_{\alpha},v_{\beta},v_{\gamma}$, then
\begin{align*}
\iota_{[v_{\alpha},v_{\beta}]} \gamma + \cp &= 
-3\ip{\alpha}\ip{\beta}\ip{\gamma}\omega +
\ip{\alpha}d\innerprodm{v_{\beta}+\beta}{v_{\gamma} + \gamma}\\ & \quad  +
\ip{\gamma}d\innerprodm{v_{\alpha}+\alpha}{v_{\beta} + \beta} +
\ip{\beta}d\innerprodm{v_{\gamma}+\gamma}{v_{\alpha} + \alpha} \\
\end{align*}
\end{lemma}
\begin{proof}
The identity $\iota_{[v_{\alpha},v_{\beta}]}=\L_{v_{\alpha}}\ip{\beta}
-\ip{\beta}\L_{v_{\alpha}}$ and Lemma \ref{calc1} imply:
\begin{align*}
\iota_{[v_{\alpha},v_{\beta}]} \gamma &=
\lie{\alpha}{\ip{\beta}}\gamma - {\ip{\beta}}\lie{\alpha}\gamma\\
&=\lie{\alpha}{\ip{\beta}}\gamma -\ip{\beta}
\left(\brac{\alpha}{\gamma} + d \ip{\alpha}\gamma\right)\\
&=\ip{\alpha}d\ip{\beta}\gamma -\ip{\beta}\ip{\gamma}\ip{\alpha}\omega 
  - \ip{\beta}d\ip{\alpha}\gamma,
\end{align*}
where the last equality follows from the definition of the bracket.

Therefore we have:
\begin{align*}
\iota_{[v_{\gamma},v_{\alpha}]} \beta  
&=\ip{\gamma}d\ip{\alpha}\beta -\ip{\alpha}\ip{\beta}\ip{\gamma}\omega 
  - \ip{\alpha}d\ip{\gamma}\beta,\\
\iota_{[v_{\beta },v_{\gamma}]}\alpha  
&=\ip{\beta}d\ip{\gamma}\alpha -\ip{\gamma}\ip{\alpha}\ip{\beta}\omega 
  - \ip{\gamma}d\ip{\beta}\alpha,
\end{align*}
and Eq.\ \ref{skew} implies
\[
\ip{\alpha}d\ip{\beta}\gamma - \ip{\alpha}d\ip{\gamma}\beta=
\ip{\alpha}d\innerprodm{v_{\beta}+\beta}{v_{\gamma} + \gamma}.
\]
The statement then follows.
\end{proof}

\begin{lemma}\label{calc3}
If $\alpha, \beta \in \ham$ with corresponding Hamiltonian vector fields
$v_{\alpha},v_{\beta}$, then
\[
\lie{\alpha}{\beta}-\lie{\beta}{\alpha}=2\brac{\alpha}{\beta} + d\innerprodm{v_{\alpha}
  + \alpha}{v_{\beta} + \beta}.
\]
\end{lemma}
\begin{proof}
Follows immediately from Lemma \ref{calc1} and Eq.\ \ref{skew}.
\end{proof}

\begin{proof}[Proof of Theorem \ref{main_thm}]
We will construct a morphism from $L_{\infty}(M,\omega)$ to $L_{\infty}(C)$.
Let
\begin{gather*}
L_{\bullet}= \cinf(M) \stackrel{d}{\to} \ham,\\
\lbrac{\cdot}{\cdot} \maps L_{\bullet} \tensor L_{\bullet}\to L_{\bullet},\\
 J_{L} \maps L_{\bullet} \tensor L_{\bullet} \tensor L_{\bullet}  \to L_{\bullet}
\end{gather*}
denote the underlying chain complex, bracket, and Jacobiator of the Lie
2-algebra $L_{\infty}(M,\omega)$. Similarly,
\begin{gather*}
L'_{\bullet} = \cinf(M) \stackrel{d}{\to} \Gamma(C),\\
\lpbrac{\cdot}{\cdot} \maps L'_{\bullet} \tensor L'_{\bullet} \to L'_{\bullet},\\
J_{L'} \maps L'_{\bullet} \tensor L'_{\bullet} \tensor L'_{\bullet} \to L'_{\bullet}
\end{gather*}
denotes the underlying chain
complex, bracket, and Jacobiator of the Lie 2-algebra
$L_{\infty}(C)$.
 
Let $s \maps TM \to C$ be the splitting. Let $\phi_{0} \maps
\ham \to \Gamma(C)$ be given by
\[
\phi_{0}(\alpha)=s(v_\alpha) +\alpha,
\]
where $v_{\alpha}$ is the Hamiltonian vector field corresponding to
$\alpha$.
Let $\phi_{1}\maps \cinf(M) \to \cinf(M)$ be the identity.
Then $\phi_{\bullet} \maps L_{\bullet} \to L'_{\bullet}$ is a chain
map, since the Hamiltonian vector field of an exact 1-form is zero.
Let $\Phi \maps \ham \tensor \ham \to \cinf(M)$ be given
by
\[
\Phi(\alpha,\beta)=-\half \innerprodm{v_{\alpha} + \alpha}{v_{\beta} + \beta}.
\]

Now we show $\Phi$ is a well-defined chain homotopy in the sense of
Definition \ref{homo}. We have
\begin{equation}\label{lpbrac_eq}
\begin{split}
\lpbrac{\phi_{0}(\alpha)}{\phi_{0}(\beta)}&=  \tcbrac{s(v_{\alpha}) +
  \alpha}{s(v_{\beta}) + \beta}\\
&=  s([v_{\alpha},v_{\beta}]) + \L_{v_{\alpha}}\beta - \L_{v_{\beta}}\alpha
 -\ip{\beta}\ip{\alpha} \omega\\
  &\quad - \half d \innerprodm{v_{\alpha} + \alpha }{v_{\beta} +
    \beta}\\
  &=s([v_{\alpha},v_{\beta}]) + \brac{\alpha}{\beta} + \half d
  \innerprodm{v_{\alpha} + \alpha }{v_{\beta} + \beta}\\
&=s([v_{\alpha},v_{\beta}]) + \lbrac{\alpha}{\beta} - d\Phi(\alpha,\beta).
\end{split}
\end{equation}
The second line above is just the definition of the twisted Courant
bracket (Eq.\ \ref{twist_bracket}), while the 
second to last line follows from Lemma \ref{calc3} and 
Def.\ \ref{semi-bracket.defn} of the bracket $\{\cdot,\cdot\}$.
By Prop.\ \ref{semi-bracket}, the Hamiltonian vector
field of $\brac{\alpha}{\beta}$ is $[v_{\alpha},v_{\beta}]$. Hence we
have:
\[
\phi_{0}(\lbrac{\alpha}{\beta}) -\lpbrac{\phi_{0}(\alpha)}{\phi_{0}(\beta)}
=d\Phi(\alpha,\beta).
\]

In degree 1, the bracket $\lbrac{\cdot}{\cdot}$ is trivial. It
follows from the definition of $\lpbrac{\cdot}{\cdot}$ that
\[
\phi_{1}(\lbrac{\alpha}{f})- 
\lpbrac{\phi_{0}(\alpha)}{\phi_{1}(f)} = -\half \innerprodp{s(v_{\alpha}) + \alpha}{df}.
\]
From Eq.\ \ref{split_innerprod}, we have
\[
\innerprodp{s(v_{\alpha}) + \alpha}{df}=\innerprodp{s(v_{\alpha}) +
  \alpha}{s(0)+ df}= \ip{\alpha}df.
\]
Therefore 
\[
\phi_{1}(\lbrac{\alpha}{f})- 
\lpbrac{\phi_{0}(\alpha)}{\phi_{1}(f)} = \Phi(\alpha,df),
\]
and similarly
\[
\phi_{1}(\lbrac{f}{\alpha})- 
\lpbrac{\phi_{1}(f)}{\phi_{0}(\alpha)} = \Phi(df,\alpha).
\]
Therefore $\Phi$ is a chain homotopy.

It remains to show the coherence condition (Eq.\ \ref{coherence} in
Definition \ref{homo}) is satisfied.  First we rewrite the Jacobiator
$J_{L'}$ using the second to last line of (\ref{lpbrac_eq}):
\begin{align*}
  J_{L'}(\phi_0(\alpha),\phi_0(\beta), \phi_0(\gamma))&=-\frac{1}{6}
  \innerprodp{\lpbrac{\phi_{0}(\alpha)}{\phi_0(\beta)}}{\phi_0(\gamma)}
  +
  \cp \\
  &=-\frac{1}{6}\innerprodp{
    s([v_{\alpha},v_{\beta}])
    +\brac{\alpha}{\beta}-d\Phi(\alpha,\beta)}{s(v_{\gamma})+\gamma}
  \\ & \quad +  \cp.
\end{align*}
From the definition of the bracket $\brac{\cdot}{\cdot}$ and the symmetric
pairing, we have
\begin{equation} \label{sub1}
  J_{L'}(\phi_0(\alpha),\phi_0(\beta), \phi_0(\gamma))=
 -\frac{1}{2}\ip{\gamma}\ip{\beta}\ip{\alpha}\omega 
-\frac{1}{6} \bigl( \iota_{[v_{\alpha},v_{\beta}]} \gamma-
\ip{\gamma}d\Phi(\alpha,\beta) + \cp \bigr).
\end{equation}
Lemma \ref{calc2} implies
\begin{equation} \label{lemma_implies}
\iota_{[v_{\alpha},v_{\beta}]} \gamma + \cp = -3 \ip{\alpha}\ip{\beta}\ip{\gamma}\omega 
- \bigl (2 \ip{\gamma}d\Phi(\alpha,\beta) + \cp \bigr),
\end{equation}
so  Eq.\ \ref{sub1} becomes
\[
 J_{L'}(\phi_0(\alpha),\phi_0(\beta), \phi_0(\gamma))=
\ip{\alpha}\ip{\beta}\ip{\gamma}\omega + \bigl(\frac{1}{2}
\ip{\gamma}d\Phi(\alpha,\beta) + \cp \bigr).
\]
By definition, $J_{L}(\alpha,\beta,\gamma)=\ip{\alpha}\ip{\beta}\ip{\gamma}\omega$.
Therefore, in this case, the left-hand side of Eq.\ \ref{coherence} is
\begin{equation}\label{LHS}
\phi_1(J_{L}(\alpha,\beta,\gamma)) - J_{L'}(\phi_0(\alpha),\phi_0(\beta), \phi_0(\gamma)) =
-\frac{1}{2}\ip{\gamma}d\Phi(\alpha,\beta) + \cp.
\end{equation}

Since the brackets and homotopy $\Phi$ are skew-symmetric, the right-hand side of Eq.\
\ref{coherence} can be rewritten as:
\begin{equation}\label{RHS}
\bigl (\Phi(\alpha,\lbrac{\beta}{\gamma}) + \cp \bigr)-
\bigl ( \lpbrac{\Phi(\alpha,\beta)}{\phi_{0}(\gamma)}  + \cp \bigr).
\end{equation}
Consider the first term in Eq.\ \ref{RHS}. The Hamiltonian vector field corresponding to
$\lbrac{\beta}{\gamma}=\brac{\beta}{\gamma}$ is
$[v_{\beta},v_{\gamma}]$. Therefore the definition of $\Phi$ implies
\[
\Phi(\alpha,\lbrac{\beta}{\gamma}) + \cp =-\frac{3}{2}\ip{\gamma}\ip{\beta}\ip{\alpha}\omega
+\frac{1}{2}\bigl( \iota_{[v_{\beta},v_{\gamma}]}\alpha + \cp\bigr).
\]
It then follows from Lemma \ref{calc2} (see Eq.\ \ref{lemma_implies}) that
\[
\Phi(\alpha,\lbrac{\beta}{\gamma}) + \cp = -  \ip{\gamma}d\Phi(\alpha,\beta) + \cp.
\]
By definition of the bracket $\lpbrac{\cdot}{\cdot}$, the second term
in Eq.\ \ref{RHS} can be written as
\[
\lpbrac{\Phi(\alpha,\beta)}{\phi_{0}(\gamma)}  + \cp = 
-\frac{1}{2} \ip{\gamma}d\Phi(\alpha,\beta) + \cp.
\]
Hence the coherence condition:
\begin{multline*}
\phi_1(J_{L}(\alpha,\beta,\gamma)) -
J_{L'}(\phi_0(\alpha),\phi_0(\beta), \phi_0(\gamma)) =\\
\Phi(\alpha,\lbrac{\beta}{\gamma}) - 
\lpbrac{\Phi(\alpha,\beta)}{\phi_{0}(\gamma)}  + \cp
\end{multline*}
is satisfied, and $(\phi_{\bullet},\Phi) \maps L_{\infty}(M,\omega)
\to L_{\infty}(C)$ is a morphism of Lie 2-algebras.
\end{proof}

We now focus on a particular sub-Lie 2-algebra of
$L_{\infty}(C)$. The following definition is due to
\v{S}evera \cite{Severa1} and is a generalization of Def.\ \ref{preserve_split_def}:
\begin{definition} \label{courant_conn_preserve_def}
Let $C$ be the exact Courant algebroid given in Prop.\ \ref{2plectic_courant}
equipped with a splitting $s \maps TM \to C$. We say 
a section $e=s(v)+\alpha$ {\bf preserves the splitting} iff 
$\forall v' \in \X(M)$
\[
\tdbrac{e}{s(v')}=s([v,v']).
\]
The subspace of sections that preserve the splitting is 
denoted as $\Gamma(C)^{s}$.
\end{definition}
Note that the twisted Dorfman bracket is used in the above definition
rather than the twisted Courant bracket. Since it satisfies the Jacobi
identity, it gives a `strict' adjoint action on sections of $C$. 
The 2-plectic analogue of Proposition \ref{preserve_lie_alg} is:
\begin{prop}
If $C$ is the exact Courant algebroid given in Proposition \ref{2plectic_courant}
equipped with the splitting $s \maps TM \to C$, then there is
a  Lie $2$-algebra
$L_{\infty}(C)^{s}=(L_{\bullet},\blankbrac, J)$ where:
\begin{itemize}
\item $L_{0}=\Gamma(C)^{s}$,
\item $L_{1}=\cinf(M)$,
\item the differential $L_{1} \stackrel{D}{\to} L_{0}$ is $D=\rho^{\ast}d$,
\item{the bracket $\sbrac{\cdot}{\cdot}$ is 
\[
[e_{1},e_{2}]= \tcbrac{e_{1}}{e_{2}} \quad  \text{in degree 0}
\]
and
\[
  [e,f]=-[f,e]=\half \innerprodp{e}{df} \quad \text{in degree 1},
\]
}
\item the Jacobiator is the linear map $J\maps \Gamma(C) ^{s} \tensor \Gamma(C) ^{s} \tensor \Gamma(C) ^{s} \to \cinf(M)$ defined by
\begin{align*} 
J(e_{1},e_{2},e_{3}) &= -T(e_{1},e_{2},e_{3})\\
&=-\frac{1}{6}
\left(\innerprodp{\tcbrac{e_1}{e_2}}{e_3} +
\innerprodp{\tcbrac{e_3}{e_1}}{e_2} \right. \\
 & \left.\quad + \innerprodp{\tcbrac{e_2}{e_3}}{e_1} \right). 
\end{align*}
\end{itemize}
\end{prop}
\begin{proof}
  Let $v'$ be a vector field on $M$. By the definition
  of the twisted Dorfman bracket (Eq.\ \ref{twist_bracket_dorf}), it follows
  that $\tdbrac{df}{s(v')}=0$ $\forall f\in \cinf(M)$. Hence the
  complex $L_{\bullet}$ is well-defined. 
 We now show that
  $\Gamma^{s}(C)$ is closed under the twisted Courant
  bracket.
Suppose $e_{1}$ and $e_{2}$ are sections preserving the
splitting. Let $e_{i}=s(v_{i})+\alpha_{i}$. Since the twisted Dorfman
bracket and the Lie bracket of vector fields satisfy the Jacobi
identity, we have:
\[
\tdbrac{\tdbrac{e_{1}}{e_{2}}}{s(v')}=  s([[v_{1},v_{2}],v']).
\]
From  Eq.\ \ref{dorfman}, we have the identity:
\[
\tcbrac{e_{1}}{e_{2}}=\tdbrac{e_{1}}{e_{2}} - \half d\innerprodp{e_{1}}{e_{2}}.
\]
Therefore:
\begin{align*}
\tdbrac{\tcbrac{e_{1}}{e_{2}}}{s(v')} &= 
\tdbrac{\tdbrac{e_{1}}{e_{2}}}{s(v')} - \half \tdbrac{d\innerprodp{e_{1}}{e_{2}}}{s(v')}\\
&=s([[v_{1},v_{2}],v']).
\end{align*}
It follows from Theorem \ref{courant_L2A} that the Lie 2-algebra
axioms are satisfied.
\end{proof}

This next result is essentially a corollary of Thm.\
\ref{main_thm}. However, it is important since
it is the 2-plectic analogue of Prop.\
\ref{lie_alg_iso}. 
\begin{prop} \label{lie_2_alg_iso}
$L_{\infty}(M,\omega)$ and $L_{\infty}(C)^{s}$ are isomorphic
as Lie 2-algebras.
\end{prop}
\begin{proof}
Recall that in Theorem \ref{main_thm} we constructed a morphism of
Lie 2-algebras given by a chain map $\phi_{\bullet} \maps L_{\infty}(M,\omega) \to
L_{\infty}(C)$:
\[
\phi_{0}(\alpha)=s(v_{\alpha}) + \alpha, \quad \phi_{1}=\id,
\]
and a homotopy $\Phi \maps \ham \tensor \ham \to \cinf(M)$:
\[
\Phi(\alpha,\beta)=-\half \innerprodm{v_{\alpha} + \alpha}{v_{\beta} + \beta}.
\]
Let $v' \in \X(M)$ and $e = s(v) + \alpha$. 
By definition of the twisted Dorfman bracket, $\tdbrac{e}{s(v')}=s[v,v']$ if and only if 
$\iota_{v'} \bigl(d\alpha +\iota_{v}\omega \bigr)=0$.
Hence a section  of $C$
preserves the splitting if and only if it lies in the image of the
chain map $\phi_{\bullet}$. Since this map is also injective, the statement follows.
\end{proof}

\section{Categorified prequantization} \label{2plectic_quant} 
In this section, we introduce a prequantization scheme for 2-plectic
manifolds, and provide a brief exposition on the higher geometric
structures which naturally appear. The relationship between the
Courant algebroid $C$ and the 2-plectic manifold $(M,\omega)$ has an
interesting interpretation when we consider the special case of
prequantized 2-plectic manifolds.  In particular, we will see that $C$
acts as the 2-plectic analogue of the Atiyah algebroid $A$ described
in Sec.\ \ref{symp_prequant}. 

\begin{definition}
A 2-plectic manifold $(M,\omega)$ admits a {\bf prequantization} iff
the cohomology class $[\omega]$ lies in the image
of the map $H^{3}(M,\Z) \to H^{3}(M,\R) \cong H^{3}_{\mathrm{DR}}(M)$. 
\end{definition}
Let $(M,\omega)$ be prequantizable. By using the maps $c \maps
H^{2}(M,D_{2}^{\bullet})\epi H^{3}(M,\Z)$ and $\kappa \maps
H^{2}(M,D_{2}^{\bullet}) \to \cOmega^{3}(M)$ discussed in Section
\ref{Deligne_sec}, we can find a Deligne class in
$H^{2}(M,D_{2}^{\bullet})$ whose 3-curvature is $\omega$. By
definition, a representative of this class defined on a good cover
$\{U_{i}\}$ is a set of 2-forms $\{B_{i} \in
\Omega^{2}(U_{i})\}$, a set of 1-forms $\{A_{ij}
\in\Omega^{1}(U_{ij}) \}$
 on double intersections, and a set of $U(1)$-valued functions
$\{g_{ijk} \maps U_{ijk} \to U(1)\}$ on triple intersections such that
\begin{equation} \label{2-cocycle}
\begin{split}
\omega&=dB_{i} ~ \text{on $U_{i}$},\\
B_{j} -B_{i} &= dA_{ij} ~ \text{on $U_{ij}$},\\
A_{jk} - A_{ik} + A_{ij} &= g^{-1}_{ijk}dg_{ijk} ~ \text{on $U_{ijk}$},\\
g_{jkl}g^{-1}_{ikl}g_{ijl}g^{-1}_{ijk}&=1 ~ \text{on $U_{ijkl}$}.
\end{split}
\end{equation}
A 2-plectic manifold equipped with such a Deligne 2-cocycle is said to
be \textbf{prequantized}.  We can use the 2-cocycle to construct the
Courant algebroid $C$ over $M$ equipped with a splitting
given locally by the 2-forms $\{B_{i}\}$.  However, the
fact that the cocycle data includes the \v{C}ech 2-cocycle $\{g_{ijk} \maps
U_{ijk} \to U(1)\}$ implies that there is an additional geometric
structure present on $M$. We would expect $C$ to be related to
this structure just as the Atiyah algebroid $A$ described in
Section \ref{symp_prequant} is related to its associated principal
bundle.

The geometric object we associate to the \v{C}ech 2-cocycle $\{g_{ijk}
\maps U_{ijk} \to U(1)\}$ is a $U(1)$-gerbe. The precise definition of
a gerbe is rather technical and can be found in Brylinski
\cite{Brylinski:1993} or Moerdijk \cite{Moerdijk:2002}.  However, in
what follows we hope to provide some intuitive geometric understanding
of these structures and motivate their proposed role in the
prequantization of 2-plectic manifolds. Additional details can be found
in Sections 5.5 and 7.2 of \cite{CLR_thesis}.

\subsection{$U(1)$-gerbes as stacks} \label{gerbes} For the purpose of
comparison, it is helpful to momentarily return to the `1-plectic'
case. Instead of associating a \v{C}ech 1-cocycle to a principal
$U(1)$- bundle $P \stackrel{\pi}{\to} M$, we can just as well
associate the cocycle to the bundle's sheaf of sections $\sh{P}$. The
sheaf $\sh{P}$ is a $\sh{U(1)}$-torsor. This means that the
sheaf of groups $\sh{U(1)}$ acts on $\sh{P}$ in such a way so that for
all $x \in M$ there exists a neighborhood $x\in U$ and an equivariant
isomorphism of sheaves $\sh{P}_{U}
\stackrel{\sim}{\to}\sh{U(1)}_{U}$. In other words, $\sh{P}$ is
locally isomorphic to the trivial torsor $\sh{U(1)}$. We recover the
\v{C}ech 1-cocycle from $\sh{P}$ in the obvious way: Choose a good
cover $\{U_{i}\}$ of $M$ such that $\sh{P}_{U_{i}}$ is isomorphic to
$\sh{U(1)}$ as a sheaf over $U_{i}$. Choose sections $\sigma_{i} \in
\sh{P}_{U_{i}}$, and consider the restrictions $\sigma_{i}
\vert_{U_{ij}}$, $\sigma_{j} \vert_{U_{ij}} \in \sh{P}(U_{ij})$. There exists sections $g_{ij} \in
\sh{P}(U_{ij})$ such that $ \sigma_{j}=\sigma_{i} \cdot g_{ij}$ on
$U_{ij}$ which obey the usual cocycle condition on $U_{ijk}$.

Now let us consider the higher analogue. Just as the $\sh{U(1)}$ torsor
$\sh{P}$ is a particular kind of sheaf, a $U(1)$-gerbe is a particular
kind of stack. A stack $\st{S}$ over $M$ is, very roughly, a
categorified sheaf over $M$. To every open set $U$ of $M$, one assigns
a groupoid $\st{S}(U)$. To every inclusion of open sets $V
\stackrel{\iota}{\to} U$, one assign a functor $ \st{S}(\iota) \maps
\st{S}(U) \to \st{S}(V)$, which pulls back, or `restricts', objects
and morphisms over $U$ to those over $V$. However,
given a composition of inclusions:
\[
\xymatrix{
        W \ar[dr]_{\iota_{VW}} \ar[rr]^{\iota_{UW}=\iota_{UV} \circ \iota_{VW}} &   & U \\
        &V \ar[ur]_{\iota_{UV}}&\\}
\]
one requires the corresponding functors $\st{S}(\iota_{UW})$ and
$\st{S}(\iota_{VW}) \circ \st{S}(\iota_{UV})$ to be equivalent via a
coherent natural isomorphism instead of being equal.  Just as the sheaf
axioms involve gluing together local sections (i.e.\ elements of 
sets), the axioms for a stack involve gluing together objects and
morphisms of groupoids. 

Perhaps the most intuitive example of a stack is the classifying stack
$\st{B}U(1)$, which assigns to every open set $U \subseteq M$ the groupoid
of principal $U(1)$-bundles over $U$. This stack has nice extra
properties. For example, for any open set $U$ and any two bundles
$P_{1},P_{2} \in \st{B}U(1)(U)$, there exists an open subset $V
\subseteq U$ such that the pullback bundles $P_{i}\vert_{V}$ are
isomorphic as objects in $\st{B}U(1)(V)$. Moreover, $V$ can be chosen
so that the automorphism groups $\Aut(P_{i} \vert_{V})$ are isomorphic
to the group of $U(1)$-valued functions $\sh{U(1)}(V)$.  Roughly,
these are the defining properties of a $U(1)$-gerbe. We may think of a
$U(1)$-gerbe $\st{G}$ over $M$ as a stack
with the additional property that there exists an open cover
$\{U_{i}\}$ of $M$ such that for all open sets $V \subseteq U_{i}$,
the groupoid $\st{G}(V)$ is equivalent (as a category) to
$\st{B}U(1)(V)$.

One obtains a \v{C}ech 2-cocycle from a $U(1)$-gerbe $\st{G}$ in the
following way: Choose a good open cover $\{U_{i}\}$ of $M$ such that
there exists objects $P_{i} \in \st{G}(U_{i})$, isomorphisms $u_{ij}
\maps P_{i} \vert_{U_{ij}} \to P_{j} \vert_{U_{ij}}$, and isomorphisms
$\Aut(P_{i}
\vert_{V}) \cong \sh{U(1)}(V)$ for all open subsets $V \subset
U_{i}$. Such a cover exists since $\st{G}$ is locally isomorphic to
$\st{B}U(1)$. By restricting these objects and isomorphisms to triple
intersections $U_{ijk}$, we obtain an automorphism $u^{-1}_{ik} u_{ij}
u_{jk}$ of $P_{k} \vert_{U_{ijk}}$. This gives a $U(1)$-valued
function $g_{ijk} \in \sh{U(1)}(U_{ijk})\cong \Aut(P_{k}
\vert_{U_{ijk}})$, which satisfies the cocycle condition on quadruple intersections.
One can show that the cohomology class given by the $g_{ijk}$ is
invariant with respect to all choices made in this construction. In
particular, $\st{B}U(1)$ gives the trivial class in
$H^{2}(M,\sh{U(1)})$. We refer the reader to Brylinski
\cite{Brylinski:1993} for the reverse construction which produces a
gerbe from a 2-cocycle. 

Since the `sections' of a $U(1)$-gerbe are locally principal
$U(1)$-bundles, they can be equipped with connections (local 1-forms) 
which give their curvatures (local 2-forms). This fact leads to the
notion of equipping the gerbe with a connection and curving.
One can proceed further and show that gerbes equipped with such structures
correspond to the aforementioned Deligne 2-cocycles (\ref{2-cocycle}). The precise
definitions of connections and curvings and their relationships to
Deligne cohomology are somewhat lengthy and technical, so we, again,
refer the reader to \cite{Brylinski:1993} for the details.

\subsection{Exact Courant algebroids as higher Atiyah algebroids} \label{higher_atiyah}
Recall that in Sec.\ \ref{symp_prequant}, we discussed how the
transitive Lie algebroid $A$ on a prequantized symplectic manifold is
isomorphic to the Atiyah algebroid associated to a principal
$U(1)$-bundle $P \to M$ equipped with a connection. Sections of the
Atiyah algebroid are $U(1)$-invariant vector fields on the total space
of the bundle. Therefore they act as infinitesimal $U(1)$-equivariant
diffeomorphisms on $P$. Prop.\ \ref{lie_2_alg_iso} implies that the
quantized Poisson algebra is the subalgebra of infinitesimal
diffeomorphisms that preserve the connection on $P$. Analogously, the
above discussion and Prop.\ \ref{2plectic_courant} imply that the
Courant algebroid $C$ on a prequantized 2-plectic manifold is
associated to a $U(1)$-gerbe $\mathcal{G} \to M$ equipped with a
connection and curving.  Furthermore, Prop.\ \ref{lie_2_alg_iso}
suggests that we interpret the Lie 2-algebra $L_{\infty}(C)^{s}$ as
the quantization of the Lie 2-algebra of `observables'
$L_{\infty}(M,\omega)$.  Clearly, these results further support the
idea that exact Courant algebroids play the role of higher Atiyah
algebroids \cite{Bressler-Chervov,Gualtieri:2007}.  However,
interpreting $L_{\infty}(C)^{s}$ as `operators' or as infinitesimal
symmetries of $\mathcal{G}$ is still a work in progress. 

One possible strategy for addressing these issues is to work 
with principal $BU(1)$ 2-bundles and Lie groupoids rather than
$U(1)$-gerbes and manifolds \cite{BaezSchreiber:2005,Bartels:2004}. 
$BU(1)$ is the one object Lie groupoid
\[
U(1) \rightrightarrows \ast
\]       
It is also an example of a strict Lie 2-group i.e.\ a Lie groupoid
that is equipped with a strict (and smooth) monoidal structure such
that all objects have inverses. The action of a Lie 2-group on a Lie
groupoid is the higher analogue of the action of a Lie group on a
manifold.  The correct morphisms between Lie groupoids are not smooth
functors, but rather `Morita maps', or `bibundles'. (See Def.\ 3.25 in
\cite{Lerman:2009}.) Since any manifold $M$ is a trivial Lie groupoid
$M \rightrightarrows M$, one can speak of a Lie groupoid morphism $M
\to BU(1)$.  By unfolding the definition of a bibundle, one can show
that such a morphism corresponds to a principal $U(1)$-bundle over
$M$. In other words, sections of the trivial principal $BU(1)$ 2-bundle
over $M$ correspond to principal $U(1)$-bundles over $M$, just as
sections of the trivial principal $U(1)$-bundle over $M$ corresponds
to $U(1)$-valued functions. Bartels \cite{Bartels:2004} showed that principal
$BU(1)$ 2-bundles are classified by the usual \v{C}ech cohomology
$H^{2}(M,\sh{U(1)})$. Given a 2-cocycle, the
corresponding $U(1)$-gerbe is the stack of sections of the
corresponding 2-bundle.  One can go further and equip a principal
$BU(1)$ 2-bundle with a `2-connection'.  These correspond to Deligne
2-cocycles \cite{BaezSchreiber:2005}.  One could try to understand how sections
of an exact Courant algebroid over a prequantized 2-plectic manifold
correspond to $BU(1)$-invariant vector fields on a principal $BU(1)$
2-bundle. This would be in complete analogy with the symplectic
case. We will, in fact, see in the next section that there is a
relationship between the Lie 2-algebra $L_{\infty}(C)^{s}$ and the Lie
2-algebra that integrates to $BU(1)$.

\section{Central extensions of Lie 2-algebras} \label{2plectic_extend}
In this section, we push the analogy between prequantization and
categorified prequantization further by constructing the 2-plectic
version of the Kostant-Souriau central extension, which we discussed
in Sec.\ \ref{symp_extend}.
First some preliminary definitions:
\begin{definition}
A Lie 2-algebra $(L_{\bullet},\blankbrac, J)$ is {\bf trivial} iff $L_{1}=0$.
\end{definition}
Any Lie algebra $\g$ gives a trivial Lie
2-algebra whose underlying  complex is
\[
0 \to \g.
\]
In particular, the Lie algebra of Hamiltonian vector fields $\Xham$ is a trivial Lie 2-algebra.
\begin{definition}
A Lie 2-algebra $(L_{\bullet},\blankbrac, J)$
is {\bf abelian} iff $\blankbrac=0$ and $J=0$.
\end{definition}
Hence an abelian Lie 2-algebra is just a 2-term chain complex. 

\begin{definition} \label{extend_def}
If $L$, $L'$, and $L''$ are Lie
2-algebras whose underlying chain complexes are
$L_{\bullet}$, $L_{\bullet}'$, and $L_{\bullet}''$, respectively, then
$L'$ is a {\bf strict extension} of $L''$ by $L$
iff there exists Lie 2-algebra morphisms
\[
(\phi_{\bullet},\Phi) \maps L \to L', \quad
(\phi'_{\bullet},\Phi') \maps L' \to L''
\]
such that 
\[
L_{\bullet} \stackrel{\phi_{\bullet}}{\to} L'_{\bullet} \stackrel{\phi'_{\bullet}}{\to} L''_{\bullet}
\]
is a short exact sequence of complexes. We say $L'$ is a {\bf strict central
  extension} of $L''$ iff
$L'$ is a strict extension of
$L''$ by $L$ and
\[
\left [ \im \phi_{\bullet}, L'_{\bullet} \right]' =0.
\]
\end{definition}
These definitions will be sufficient for our discussion here. However,
they are, in general, too strict. For example, one can have homotopies
between morphisms between Lie 2-algebras, and therefore we should
consider sequences that are only exact up to homotopy as
``exact''. 
In what
follows, by an extension we mean a strict extension in the sense of
Def. \ref{extend_def}.

We would like to understand how $L_{\infty}(M,\omega)$ is a central extension
of $\Xham$ as a Lie 2-algebra. Our first two results are quite general
and hold for any 2-plectic manifold $(M,\omega)$.
\begin{prop} \label{gen_extend}
If $(M,\omega)$ is a 2-plectic manifold, then
the Lie 2-algebra $L_{\infty}(M,\omega)$ is a central extension of the
trivial Lie 2-algebra $\Xham$ by the abelian Lie 2-algebra
\[
\cinf(M) \stackrel{d}{\to} \cOmega^{1}(M),
\]
consisting of smooth functions and closed 1-forms.
\end{prop}
\begin{proof}
Consider the following short exact sequence of complexes:
\begin{equation}\label{ses}
    \xymatrix{
        \cOmega^1(M) \ar[r]^{\jmath} & \ham \ar[r]^{p}  & \Xham\\
        \cinf(M) \ar[u]^{d}\ar[r]^{\id}  & \cinf(M) \ar[u]^{d}
        \ar[r]& 0 \ar[u]}
\end{equation}
The map $\jmath \maps \cOmega^{1}(M) \to \ham$ is the inclusion, and
\[
p \maps \ham \to \Xham, \quad p(\alpha)=v_{\alpha}
\]
takes a Hamiltonian 1-form to its corresponding vector field. It
follows from Prop. \ref{semi-bracket} that $p$ preserves the bracket. In fact,
all of the horizontal chain maps give strict Lie 2-algebra morphisms
(i.e.\ all homotopies are trivial). The Hamiltonian vector field
corresponding to a closed 1-form is zero. Thus,
if $\alpha$ is closed, then for all $\beta \in \ham$
we have $\sbrac{\alpha}{\beta}_{L_{\infty}(M,\omega)}=\brac{\alpha}{\beta}=0$.
Hence $L_{\infty}(M,\omega)$ is a central extension of $\Xham$. 
\end{proof}

\begin{prop} \label{L2A_extend_1} 
Let $(M,\omega)$ be a 2-plectic manifold. Given $x\in M$, there
is a Lie 2-algebra  $L_{\infty}(\Xham,x)=(L_{\bullet},\blankbrac,J_{x})$ where
\begin{itemize}
\item{$L_{0}=\Xham$,}
\item{$L_{1}=\R$,}
\item{the differential $L_{1} \stackrel{d}{\to} L_{0}$ is trivial ($d=0$),}
\item{the bracket $\blankbrac$ is the Lie bracket on $\Xham$
in degree 0 and trivial in all other degrees}
\item{the Jacobiator is the linear map 
\[
J_{x} \maps \Xham \tensor
    \Xham \tensor \Xham \to \R
\] 
defined by
\[
J_{x}(v_{1},v_{2},v_{3})= \ip{1}\ip{2}\ip{3}\omega\vert_{x}.
\]
}
\end{itemize}
Moreover, $J_{x}$ is a 3-cocycle in the Chevalley-Eilenberg cochain
complex \linebreak $\Hom(\Lambda^{\bullet}\Xham, \R)$.
\end{prop} 
\begin{proof}
  We have a bracket defined on a complex with trivial differential
  that satisfies the Jacobi identity ``on the nose''. Hence to show
  $L_{\infty}(\Xham,x)$ is a Lie 2-algebra, it sufficient to show that the
  Jacobiator $J_{x}(v_{1},v_{2},v_{3})$ satisfies Eq.\ \ref{big_J} in
  Def.\ \ref{L2A} for $x \in M$. This follows immediately from Thm.\
  \ref{semistrict}. The classification theorem of Baez and Crans
(Thm. 55 in \cite{HDA6}) implies that $J_{x}$ satisfying Eq.\
\ref{big_J} in the definition of a Lie 2-algebra is equivalent to
$J_{x}$  being a 3-cocycle with values in the trivial representation.
\end{proof}

Recall that in the symplectic case, if the manifold is connected, then 
the Poisson algebra is a central extension of the Hamiltonian vector
fields by the Lie algebra $\u(1) \cong \R$. The categorified analog of
the Lie algebra $\u(1)$ is the abelian Lie 2-algebra $b\u(1)$ whose
underlying chain complex is simply
\[
\R \to 0.
\]
This Lie 2-algebra integrates to the Lie 2-group $BU(1)$
discussed in Sec. \ref{higher_atiyah}.
It is natural to suspect that, under
suitable topological conditions, the abelian Lie algebra $\cinf(M)
\stackrel{d}{\to} \cOmega^{1}(M)$ introduced in Prop. \ref{gen_extend}
is related to $b\u(1)$. 

Let us first assume that the 2-plectic manifold is connected. Note
that the Jacobiator $J_{x}$ of the Lie 2-algebra $L_{\infty}(\Xham,x)$ 
introduced in Prop.\ \ref{L2A_extend_1} depends
explicitly on the choice of $x \in M$. However, if $M$ is connected,
then the cohomology class $J_{x}$ represents as a 3-cocycle does not
depend on $x$. This fact has important implications for $L_{\infty}(\Xham,x)$:

\begin{prop} \label{L2A_extend_2}
If $(M,\omega)$ is a connected 2-plectic manifold and 
$J_{x}$ is the 3-cocycle given in Prop.\ \ref{L2A_extend_1}, then
the cohomology class $[J_{x}]\in H^{3}_{\mathrm{CE}}(\Xham,\R)$
is independent of the choice of $x \in M$. Moreover, given any other point
$y \in M$, the Lie 2-algebras $L_{\infty}(\Xham,x)$ and
$L_{\infty}(\Xham,y)$ are quasi-isomorphic.
\end{prop}
\begin{proof}
To prove that $[J_{x}]$ is independent of $x$, we use a construction
similar to the proof of Prop.\ 4.1 in \cite{Brylinski:1990}. 
The Chevalley-Eilenberg differential
\[
\delta \maps \Hom(\Lambda^{n}
\Xham, \R)\to \Hom(\Lambda^{n+1} \Xham, \R)
\]
 is defined by
\[
(\delta c)(v_{1},\ldots,v_{n+1}) = \sum_{1 \leq i < j \leq n}
(-1)^{i+j}c([v_{i},v_{j}],v_{1},\cdots,
\hat{v}_{i},\cdots,\hat{v}_{j},\ldots,v_{n+1}).
\]
Note that if $c$ is an arbitrary 2-cochain then 
\[
(\delta c) (v_{\alpha},v_{\beta},v_{\gamma})   =-c([v_{\alpha},v_{\beta}],v_{\gamma}) +
c([v_{\alpha},v_{\gamma}],v_{\beta}) - c([v_{\beta},v_{\gamma}],v_{\alpha}).
\]
Now let $y \in M$. Let $\Gamma \maps [0,1] \to M$ be a path from $x$ to
$y$. Given $v_{\alpha},v_{\beta} \in \Xham$, define
\[
c(v_{\alpha},v_{\beta}) = \int_{\Gamma} \omega(v_{\alpha},v_{\beta},\cdot).
\]
Clearly, $c$ is a 2-cochain. We claim
\[
J_{y}(v_{\alpha},v_{\beta},v_{\gamma})
-J_{x}(v_{\alpha},v_{\beta},v_{\gamma}) = (\delta c) (v_{\alpha},v_{\beta},v_{\gamma}) 
\]
From part 3 of Prop. \ref{semi-bracket}, we have:
\[
d\ip{\alpha}\ip{\beta}\ip{\gamma}\omega = 
\brac{\alpha}{\brac{\beta}{\gamma}} -
    \brac{\brac{\alpha}{\beta}}{\gamma} 
    -\brac{\beta}{\brac{\alpha}{\gamma}}.
\]
By definition of the bracket $\brac{\cdot}{\cdot}$, this implies
\[
d\ip{\alpha}\ip{\beta}\ip{\gamma}\omega = 
-\omega([v_{\alpha},v_{\beta}],v_{\gamma},\cdot)  
+\omega([v_{\alpha},v_{\gamma}],v_{\beta},\cdot)  
- \omega([v_{\beta},v_{\gamma}],v_{\alpha},\cdot).  
\]
Integrating both sides of the above equation gives
\begin{align*}
\int_{\Gamma} d\ip{\alpha}\ip{\beta}\ip{\gamma}\omega  &=
J_{y}(v_{\alpha},v_{\beta},v_{\gamma})
-J_{x}(v_{\alpha},v_{\beta},v_{\gamma})\\ 
&= -\int_{\Gamma}\omega([v_{\alpha},v_{\beta}],v_{\gamma},\cdot)  
+ \int_{\Gamma} \omega([v_{\alpha},v_{\gamma}],v_{\beta},\cdot)  
- \int_{\Gamma}\omega([v_{\beta},v_{\gamma}],v_{\alpha},\cdot)\\
&= (\delta c) (v_{\alpha},v_{\beta},v_{\gamma}).
\end{align*}

It follows from Thm.\ 57 in Baez and Crans \cite{HDA6} that
$[J_{x}]=[J_{y}]$ implies $L_{\infty}(\Xham,x)$ and $L_{\infty}(\Xham,y)$ 
are quasi-isomorphic (or `equivalent' in their terminology). 
\end{proof}

Now we impose further conditions on our 2-plectic manifold.  From here
on, we assume $(M,\omega)$ is 1-connected (i.e.\ connected and simply
connected). This is the 2-plectic analogue of the requirement that the
symplectic manifold in Sec.\ \ref{symp_extend} be connected. It will
allow us to construct several elementary, yet interesting,
quasi-isomorphisms of Lie 2-algebras.

\begin{prop} \label{L2A_extend_3}
If $M$ is a 1-connected manifold, then the abelian
Lie 2-algebra $\cinf(M) \stackrel{d}{\to} \cOmega^{1}(M)$ is
quasi-isomorphic to $b\u(1)$.
\end{prop}
\begin{proof}
Let $x \in M$. The chain map
\[
    \xymatrix{
       \cOmega^{1}(M) \ar[r] & 0 \\
         \cinf(M) \ar^{d}[u]\ar[r]^{~\mathrm{ev}_{x}}  & \R \ar[u]
}
\]
is a quasi-isomorphism.
\end{proof}

\begin{prop} \label{L2A_extend_4} 
If $(M,\omega)$ is a 1-connected
  2-plectic manifold and $x \in M$, then the Lie 2-algebras
  $L_{\infty}(M,\omega)$ and $L_{\infty}(\Xham,x)$ are quasi-isomorphic.
\end{prop} 

\begin{proof}
  We construct a quasi-isomorphism from $L_{\infty}(M,\omega)$ to
  $L_{\infty}(\Xham,x)$.  There is a chain map
\[
    \xymatrix{
        \ham \ar[r]^{p} & \Xham\\
        \cinf(M) \ar[u]^{d} \ar[r] ^{\mathrm{ev}_{x}}  & \R \ar[u]_{0}  
}
\]
with $\mathrm{ev}_{x}(f)=f(x)$ and $p(\alpha)=v_{\alpha}$. Since
$p$ preserves the bracket, we take $\Phi$ in Def. \ref{homo} to be
the trivial homotopy. Eq.\ \ref{coherence} holds
since:
\[
\mathrm{ev}_{x}(\omega(v_{\gamma},v_{\beta},v_{\alpha}))=J_{x}(v_{\alpha},v_{\beta},v_{\gamma}),
\]
and therefore we have constructed a Lie 2-algebra morphism. Since $M$
is connected, the homology of the complex $\cinf(M) \stackrel{d}{\to}
\ham$ is just $\R$ in degree 1 and $\ham/d\cinf(M)$ in degree 0.  The
kernel of the surjective map $p$ is the space of closed 1-forms, which
is $d\cinf(M)$ since $M$ is simply connected.
\end{proof} 

We can summarize the results given in Props.\ \ref{gen_extend}
\ref{L2A_extend_1} \ref{L2A_extend_3}, and \ref{L2A_extend_4} with
the following commutative diagram:
\[
\xymatrix{ \cOmega^{1}(M) \ar @{~>}[rd] \ar[rr]^{\jmath} && \ham
  \ar @{~>}[rd]_{p} \ar[rr]^{p} && \Xham \ar[rd]^{\id}\\
  &0 \ar[rr]  && \Xham \ar[rr] && \Xham \\
  \cinf(M) \ar'[r][rr] \ar[uu]^{d} \ar @{~>}[rd]_{\mathrm{ev}_{x}} && \cinf(M) \ar'[u][uu]
  \ar @{~>}[rd]_{\mathrm{ev}_{x}} \ar'[r][rr] && 
  0 \ar'[u][uu]\ar[rd] \\
  & \R \ar[uu] \ar[rr] && \R \ar[rr] \ar[uu] && 0 \ar[uu] }
\]
The back of the diagram shows $L_{\infty}(M,\omega)$ as the central extension of
the trivial Lie 2-algebra $\Xham$. The front shows 
$L_{\infty}(\Xham,x)$ as a central extension of $\Xham$ by
$b\u(1)$. The morphisms going from back to front are all
quasi-isomorphisms. Thus we have the 2-plectic analogue of the
Kostant-Souriau central extension:
\begin{corollary}
If $(M,\omega)$ is a 1-connected 2-plectic manifold, then
$L_{\infty}(M,\omega)$ is quasi-isomorphic to a central extension of
the trivial Lie 2-algebra $\Xham$ by $b\u(1)$.
\end{corollary}
Also, from Prop.\ \ref{lie_2_alg_iso} we know that
$L_{\infty}(M,\omega)$ is isomorphic to the Lie 2-algebra
$L_{\infty}(C)^{s}$ consisting of sections of the
Courant algebroid $C$ which preserve a chosen splitting $s
\maps TM \to C$. Therefore:
\begin{corollary}
If $(M,\omega)$ is a 1-connected 2-plectic manifold, then
$L_{\infty}(C)^{s}$ is quasi-isomorphic to a central extension of
the trivial Lie 2-algebra $\Xham$ by $b\u(1)$.
\end{corollary}
A comparison of the above corollary to
the results discussed in Sec.\ \ref{symp_extend} suggests that $L_{\infty}(C)^{s}$ be
interpreted as the quantization of $L_{\infty}(M,\omega)$ with
$b\u(1)$ giving rise to the quantum phase.

Finally, note that a splitting of the short exact sequence of complexes 
\[
    \xymatrix{
        0 \ar[r] & \Xham \ar[r]^{\id}  & \Xham\\
        \R \ar[u]\ar[r]^{\id}  & \R \ar[u]^{0} 
        \ar[r]& 0 \ar[u]}
\]
is the identity map in degree 0 and the trivial map in degree
1. Obviously the splitting preserves the bracket but does not preserve
the Jacobiator. Indeed, the failure of the splitting to be a strict
Lie 2-algebra morphism between $\Xham$ and $L_{\infty}(\Xham,x)$ is
due to the presence of the 3-cocycle $J_{x}$.

\section{Conclusion}
Let us summarize the main points of the previous sections:
If $(M,\omega)$
is a 0-connected, prequantized symplectic manifold, then 
there exists a principal $U(1)$-bundle over $M$ 
equipped with a connection whose curvature is $\omega$, and a corresponding Atiyah
algebroid $A \to M$ equipped with a splitting such that
the Lie algebra of sections of $A$ which preserve the splitting
is isomorphic to a central extension of the Lie algebra of Hamiltonian
vector fields:
\[
\u(1) \to \cinf(M) \to \Xham.
\]
This central extension gives a cohomology class in
$H^{2}_{\mathrm{CE}}(\Xham,\R)$ which can be represented by the
symplectic form evaluated at a point in $M$.

Analogously, if $(M,\omega)$
is a 1-connected, prequantized 2-plectic manifold, then 
there exists a $U(1)$-gerbe over $M$ 
equipped with a connection and curving whose 3-curvature is $\omega$, 
and a corresponding exact Courant algebroid $C \to M$
equipped with a splitting such that
the Lie 2-algebra of sections of $C$ which preserve the splitting
is quasi-isomorphic to a central extension of the (trivial) Lie 2-algebra of Hamiltonian
vector fields:
\[
b\u(1) \to L_{\infty}(\Xham) \to \Xham.
\]
This central extension gives a cohomology class in
$H^{3}_{\mathrm{CE}}(\Xham,\R)$ which can be represented by the
2-plectic form evaluated at a point in $M$.

In future work, we will develop this analogy further in order to
obtain a categorified geometric quantization procedure for
2-plectic manifolds. In doing so, it is likely that we will make
contact with related areas of interest including the representation
theory of loop groups and extended topological quantum field theories
(TQFTs). Such a procedure would also provided new insights into the
theory of Courant algebroids.

However, there are several open problems in prequantization that we
are currently addressing as we set our sights on full quantization. We
have mentioned some of these throughout the text, and we summarize
them here:
\begin{itemize}
\item{For every principal $U(1)$ bundle with connection, there is an
    associated hermitian line bundle with connection, whose
    global sections give a Hilbert space. What is the
    corresponding geometric object for a $U(1)$-gerbe equipped with a
    connection and curving? (One possible answer is described in Sec.\
    5.5 of \cite{CLR_thesis}.)
} 
\item{Sections of the Atiyah algebroid on a prequantized symplectic
    manifold are operators on this Hilbert
    space. How do sections of the Courant algebroid on a prequantized
    2-plectic manifold act as operators on the higher analogue of this
    Hilbert space?
}
\item{Sections of the Atiyah algebroid are 
    infinitesimal $U(1)$-equivariant symmetries of the corresponding
    principal $U(1)$-bundle. Integration gives elements of the gauge group
    i.e.\ equivariant diffeomorphisms of the principal bundle. 
    How can we understand sections of the Courant algebroid on a
    prequantized 2-plectic manifold as infinitesimal symmetries of
    the corresponding $U(1)$-gerbe?
}

\end{itemize}

\section{Acknowledgments}
We thank John Baez, Maarten Bergvelt,
Yael Fregier, Dmitry Roytenberg, Urs Schreiber,
Jim Stasheff and Marco Zambon for helpful conversations. We also
thank the organizers and participants of the G\"{o}ttingen Mathematics Institute
workshop: ``Higher Structures in Topology and Geometry IV'' for
their helpful comments and questions.


\begin{thebibliography}{99}
\bibitem{HDA6}
J.\ Baez and A.\ Crans, Higher-dimensional algebra VI: Lie 2-algebras, 
\textsl{Theory Appl.\ Categ.} {\bf 12} (2004), 492--528.  Also available as 
\href{http://arxiv.org/abs/math/0307263}{arXiv:math/0307263}.

\bibitem{Baez:2008bu}
J.\ Baez, A.\ Hoffnung, and C.\ Rogers, Categorified symplectic
geometry and the classical string, {\sl Comm.\ Math.\ Phys.} \textbf{293} (2010), 701--715.
Also available as \href{http://arxiv.org/abs/0808.0246}{arXiv:0808.0246}.



\bibitem{Baez:2009uu}
J.\ Baez and C.\ Rogers, Categorified symplectic geometry and the
string Lie 2-algebra, \textsl{Homology Homotopy Appl.} \textbf{12}
(2010), 221--236. Also available as \href{http://arxiv.org/abs/0901.4721}{arXiv:0901.4721}.

\bibitem{BaezSchreiber:2005} J.\ Baez and U.\ Schreiber, Higher
gauge theory, in {\sl Categories in Algebra, Geometry and Mathematical 
Physics}, eds.\ A.\ Davydov {\it et al}, {\sl Contemp. Math.} 
{\bf 431}, AMS, Providence, Rhode Island, 2007, pp.\ 7--30.
Also available as \href{http://arxiv.org/abs/math/0511710}{\texttt arXiv:math/0511710}.

\bibitem{Bartels:2004} T.\ Bartels, Higher gauge theory: 2-bundles,
  available as \href{http://arxiv.org/abs/math/0410328}{arXiv:math/0410328}.

\bibitem{Bressler-Chervov}P.\ Bressler and A.\ Chervov, Courant
  algebroids, \textsl{J.\ Math.\ Sci.\ (N.Y.)} \textbf{128} (2005), 3030--3053.
Also available as \href{http://arxiv.org/abs/hep-th/0212195}{arXiv:hep-th/0212195}.

\bibitem{Brylinski:1990}
J.-L.~Brylinski, Noncommutative Ruelle-Sullivan type currents, in 
\textsl{The Grothendieck Festschrift, Vol.\ I}, eds.\ P.\ Cartier
\textit{et al}, \textsl{Progr.\ Math.} \textbf{86} (1990), 477--498. 

\bibitem{Brylinski:1993}
J.-L.~Brylinski, {\sl Loop Spaces, Characteristic Classes and Geometric
Quantization}, Birkhauser, Boston, 1993.


\bibitem{CdS-Weinstein:1999}
A.\ Cannas da Silva and A.\ Weinstein, 
\textsl{Geometric Models for Noncommutative Algebras}, 
Berkeley Mathematics Lecture Notes \textbf{10}, Amer.\ Math.\ Soc.,
Providence, 1999.

  \bibitem{Cantrijn:1999} F.\ Cantrijn, A.\ Ibort, and M. De Leon, On
the geometry of multisymplectic manifolds, {\sl J.\ Austral.\ Math.\
  Soc.\ (Series A)} \textbf{66} (1999), 303--330.

\bibitem{Carey:2004} A.\ L.\ Carey, S.\ Johnson, and M.\ K.\ Murray,
  Holonomy on D-branes, \textsl{J.\ Geo.\ Phys.} \textbf{52} (2004),
  186--216.  Also available as
  \href{http://arxiv.org/abs/hep-th/0204199}{arXiv:hep-th/0204199v3}.

\bibitem{Courant} T.\ Courant, Dirac manifolds, \textsl{Trans.\ Amer.\
    Math.\ Soc.}  \textbf{319} (1990), 631--661.

\bibitem{Dorfman1} I.\ Dorfman, \textsl{Dirac structures and
    integrability of nonlinear evolution equations}, Nonlinear
  Science: Theory and Applications. John Wiley \& Sons, Ltd.,
  Chichester, 1993.

\bibitem{Dorfman2} I.\ Ya.\ Dorfman, Dirac structures of integrable
  evolution equations, \textsl{Phys.\ Lett.\ A.} \textbf{125} (1987), 240--246.





\bibitem{Gualtieri:2007} M.\ Gualtieri, Generalized complex geometry,
  \textsl{Ann.\ of Math.\ (2)} \textbf{174} (2011), 74--123.  Also
  available as \href{http://arxiv.org/abs/math/0703298}{arXiv:math/0703298}.

\bibitem{Helein} F.\ H\'{e}lein, Hamiltonian formalisms for
  multidimensional calculus of variations and perturbation theory, in
  \textsl{Noncompact Problems at the Intersection of Geometry}, eds.\
  A.\ Bahri \textit{et al}, AMS, Providence, Rhode Island, 2001, pp.\
  127--148. Also available as
  \href{http://arxiv.org/abs/math-ph/0212036}{arXiv:math-ph/0212036}.


\bibitem{Hitchin:2004ut} N.\ Hitchin, Generalized Calabi-Yau
  manifolds, \textsl{Quart.\ J.\ Math.\ Oxford Ser.} \textbf{54}
  (2003), 281--308. Also available as
  \href{http://arxiv.org/abs/math/0209099v1}{arXiv:math/0209099v1}.

\bibitem{Kijowski} J.\ Kijowski, A finite-dimensional canonical
formalism in the classical field theory, \textsl{Commun.\ Math.\
Phys.} \textbf{30} (1973), 99--128.

\bibitem{Kostant:1970}
B.\ Kostant, Quantization and unitary representations, \textsl{Lecture Notes in Math.}
\textbf{170} (1970), 87--208.

\bibitem{Lada-Markl} T.\ Lada and M.\ Markl, Strongly homotopy Lie
  algebras, \textsl{Comm.\ Algebra.} \textbf{23} (1995),
  2147--2161. Also available as
  \href{http://arxiv.org/abs/hep-th/9406095v1}{arXiv:hep-th/9406095}.


\bibitem{LS} T.\ Lada and J.\ Stasheff, Introduction to sh Lie
  algebras for physicists, {\sl Int.\ Jour.\ Theor.\ Phys.}
  \textbf{32} (7) (1993), 1087--1103.  Also available as
  \href{http://arxiv.org/abs/hep-th/9209099}{hep-th/9209099}.

\bibitem{Lerman:2009}
E.\ Lerman, Orbifolds as stacks?, \textsl{Enseign. Math. (2)}
\textbf{56} (2010), 315--363. Also available as \href{http://arxiv.org/abs/0806.4160}{arXiv:0806.4160v2}.

\bibitem{Liu:1997}
Z-J.\ Liu, A.\ Weinstein, and P.\ Xu, Manin triples for Lie
bialgebroids, \textsl{J.\ Differential Geometry} \textbf{45} (1997), 547--574.

\bibitem{Mackenzie:1987}
K.\ Mackenzie, \textsl{Lie groupoids and Lie algebroids in differential geometry},
London Math.\ Soc.\ Lecture Note Ser. {\bf 124},
Cambridge U.\ Press, Cambridge, 1987.

\bibitem{Moerdijk:2002} I.\ Moerdijk, Introduction to the language of
  stacks and gerbes. Available as
  \href{http://arxiv.org/abs/math/0212266}{arXiv:math/0212266v1}.

\bibitem{CLR_thesis} C.\ Rogers, \textsl{Higher Symplectic Geometry}, Ph.D.\
  thesis, UC Riverside, 2011. Also available as
\href{http://arxiv.org/abs/1106.4068}{arXiv:1106.4068}.

\bibitem{Rogers:2010nw} C.\ Rogers, $L_{\infty}$-algebras from
  multisymplectic geometry, \textsl{Lett.\ Math.\ Phys.} \textbf{100} (2012), 29--50.
  Also available as \href{http://arxiv.org/abs/1005.2230}{arXiv:1005.2230}.

\bibitem{RomanRoy:2005en}
N.\ Rom\'{a}n-Roy, Multisymplectic Lagrangian and Hamiltonian
formalisms of classical field theories, \textsl{SIGMA} \textbf{5}
(2009), 100, 25 pages. Also available as 
\href{http://arxiv.org/abs/math-ph/0506022}{arXiv:math-ph/0506022}.

\bibitem{Roytenberg-Weinstein}D.\ Roytenberg and A.\ Weinstein,
  Courant algebroids and strongly homotopy Lie algebras, 
 \textsl{Lett.\ Math.\ Phys.}  \textbf{46} (1998), 81--93. Also
   available as \href{http://arxiv.org/abs/math/9802118}{arXiv:math/9802118}.

\bibitem{Roytenberg_thesis}D.\ Roytenberg, \textsl{Courant Algebroids,
    Derived Brackets and Even Symplectic Supermanifolds}, Ph.D.\
  thesis, UC Berkeley, 1999. Also available as
\href{http://arxiv.org/abs/math/9910078}{arXiv:math/9910078}.

\bibitem{Roytenberg_graded} D.\ Roytenberg, On the structure of graded
  symplectic supermanifolds and Courant algebroids, in
  \textsl{Quantization, Poisson Brackets and Beyond}, ed.\ T.\
  Voronov, \textsl{Contemp. Math.}, \textbf{315}, AMS, Providence, RI,
  2002, pp.\ 169--185. Also available as
  \href{http://arxiv.org/abs/math/0203110}{arXiv:math/0203110}.

\bibitem{Roytenberg_L2A} D.\ Roytenberg, On weak Lie 2-algebras, 
in: P.\ Kielanowski \textit{et al} (eds.)\ XXVI Workshop on Geometrical Methods in Physics.
AIP Conference Proceedings \textbf{956}, pp. 180-198. 
American Institute of Physics, Melville (2007).
Also available as \href{http://arxiv.org/abs/0712.3461}
{\texttt arXiv:0712.3461}.


\bibitem{Schreiber:2005}
U.\ Schreiber, {\sl From Loop Space Mechanics to Nonabelian Strings},
Ph.D.\ thesis, Universit\"at Duisburg-Essen, 2005.  Also available as
\href{http://arxiv.org/abs/hep-th/0509163}{\texttt arXiv:hep-th/0509163}.



\bibitem{Severa1} P.\ \v{S}evera, Letter to Alan Weinstein, available
  at \url{http://sophia.dtp.fmph.uniba.sk/~severa/letters/}.

\bibitem{Severa-Weinstein} P.\ \v{S}evera and A.\ Weinstein, Poisson
  geometry with a 3-form background, \textsl{Prog.\ Theor.\ Phys.\
    Suppl.} \textbf{144} (2001), 145--154. Also available as
  \href{http://arxiv.org/abs/math/0107133}{arXiv:math/0107133}.


\bibitem{Souriau:1967}
J.-M.~ Souriau, Quantification g\'{e}om\'{e}trique: Applications, 
\textsl{Ann.\ Inst.\ H.\ Poincar\'{e} Sect.\ A (N.S.)} \textbf{6}
(1967), 311--341.



\end{thebibliography}
\end{document}